\documentclass[10pt,reqno,A4paper]{amsart}

\usepackage{amsmath,amssymb,amsfonts,amsthm}
\usepackage{moreverb,rotating,graphics}
\usepackage[T1]{fontenc}
\usepackage[latin1]{inputenc}
\usepackage{epsfig}
\usepackage[english]{babel} 
\usepackage{dsfont}
\usepackage{color}
\usepackage{stmaryrd}
\usepackage{psfrag}
\usepackage{tikz}

\newtheorem{theorem}{Theorem}[section]
\newtheorem{proposition}[theorem]{Proposition}
\newtheorem{remark}[theorem]{Remark}
\newtheorem{lemma}[theorem]{Lemma}

\linethickness{0.3mm}
\usepackage{setspace}
\usepackage{verbatim}

\def\tr{{\rm Tr \,}}

\def\B{{\mathcal B}}
\def\S{{\mathfrak{S} }}

\def\L{{\mathcal L}}

\def\G{{\mathcal G}}

\def\E{{\mathcal E}}
\def\K{{\mathcal K}}
\def\C{{\mathcal C}}
\def\O{{\Omega}}
\def\o{{\omega}}
\def\NN{{\mathbb N}}

\def\RR{{\mathbb R}}
\def\RRd{{\mathbb R ^d}}
\def\ZZd{{\mathbb Z ^d}}
\def\CC{{\mathbb C}}
\def\PP{{\mathbb P}}
\def\EE{{\mathbb E}}
\def\1{{\mathds{1}}}
\def\D{{\mathcal D}}

\def\ii{{\infty }}
\def\per{{\rm per }}
\def\unif{ \text{\rm  unif }  }

\def\cvL{{\operatorname*{\longrightarrow}_{L\rightarrow\ii}}}
\def\cL{{{L\rightarrow\ii}}}

\newcommand{\norm}[1]{\left\| #1\right\|}
\newcommand{\set}[1]{\left\{ #1\right\}}
\newcommand{\bra}[1]{ \left( #1\right)}
\newcommand{\av}[1]{\left| #1\right|}
\newcommand{\com}[1]{\left[ #1\right]}
\renewcommand{\phi}{\varphi}

\newcommand{\tv}[1]{ {\rm \underline{Tr}} \,\left( #1\right)}

\title[Short-range quantum crystals with defects]{The reduced Hartree-Fock model for short-range quantum crystals with defects}

\author{Salma Lahbabi}
\address{{ CNRS \& Laboratoire de Mathématiques (UMR 8088), Université de Cergy-Pontoise}\\
{\footnotesize 95000 Cergy-Pontoise Cedex, France}
\\
{CERMICS, \'Ecole Nationale des Ponts et Chaussées (Paristech)}\\
{\footnotesize  \& INRIA (Micmac Project), 6-8 Av. Blaise Pascal, 77455 Champs-sur-Marne, France }
}
\date{\today} 

\begin{document}
\maketitle

\begin{abstract}
In this article, we consider quantum crystals with defects in the reduced Hartree-Fock framework. 
The nuclei are supposed to be classical particles arranged around a reference periodic configuration. 
The perturbation is assumed to be small in amplitude, but need not be localized in a specific region of space or have any spatial invariance. Assuming Yukawa interactions, we prove the existence of an electronic ground state, solution of the self-consistent field equation. Next, by studying precisely the decay properties of this solution for local defects, we are able to expand the density of states  of the nonlinear Hamiltonian of a system with a random perturbation of Anderson-Bernoulli type, in the limit of low concentration of defects. One important step in the proof of our results is the analysis of the dielectric response of the  crystal to an effective charge perturbation. 
\end{abstract}


\section{Introduction}

In solid state physics and materials science, the presence of defects in materials induces many interesting properties, such as Anderson localization
and leads to many applications such as doped semi-conductors
The mathematical modeling and the numerical simulation of the electronic structure of these materials is a challenging task, as we are in the presence of infinitely many interacting particles.

\medskip

The purpose of this paper  is to construct the state of the quantum electrons of a mean-field crystal, in which the nuclei are classical particles arranged around a reference periodic configuration. We work with the assumption that the nuclear distribution is close to a chosen periodic arrangement \emph{locally}, but the perturbation need not be localized in a specific region of space and it also need not have any spatial invariance. To our knowledge, this is the first result of this kind for Hartree-Fock type models for quantum crystals, with short-range interactions. By studying precisely the behavior of our solution, we are then able to expand the density of states  of the Hamiltonian of the system in the presence of a random perturbation of Anderson-Bernoulli type, in the limit of low concentration of defects, that is when the Bernoulli parameter $p$ tends to zero. The state of the random crystal and the mean-field Hamiltonian were recently constructed in~\cite{CaLaLe-12}. 
Our small-$p$ expansion is the nonlinear equivalent of a previous result by Klopp~\cite{Klopp-95} in the linear case.

\medskip

The mean-field model we consider in this paper is the reduced Hartree-Fock model~\cite{Solovej}, also called the Hartree model in the physics literature. It is obtained from the generalized Hartree-Fock model~\cite{LieSim-77} by removing the exchange term. As the 
Coulomb interaction is long-range, it is a difficult mathematical question to describe infinite systems interacting through the Coulomb potential. In the following, we assume that all the particles interact through Yukawa potential of parameter $m>0$. In fact, we can assume any reasonable short-range potential, but we concentrate on the Yukawa interaction in dimension $d\in\set{1,2,3}$ for simplicity. We consider systems composed of infinitely many classical nuclei distributed over the whole space and infinitely many electrons. 

We start by recalling the definition of the reduced Hartree-Fock (rHF) model for a finite system composed of a set of nuclei having a density of charge $\nu_{\rm nuc}$ and $N$ electrons. The electrons are described by the $N$-body wave-function (called a Slater determinant)
$$
\psi(x_1,\cdots,x_N)=\frac{1}{\sqrt{N!}}\text{det}( \phi_j(x_i)),
$$
where the functions $\phi_i\in L^2(\RRd)$ satisfy $ \langle \phi_i,\phi_j\rangle=\delta_{ij}$. The rHF equations then read
\begin{equation}\label{eq:SCF_finite}
\left\{ 
\begin{array}{l}
  H\phi_i=\lambda_i\phi_i\\[0,2cm]
\displaystyle H=-\frac12 \Delta +V\\[0,2cm]
-\Delta V+m^2V= \av{S^{d-1}}\bra{\rho_\psi-\nu_{\rm nuc}}
 \end{array}
\right. \quad \forall\, 1\leq i\leq N,
\end{equation}
where $\rho_\psi(x)=\sum_{i=1}^N\av{\phi_i(x)}^2$ and $\lambda_1,\cdots, \lambda_N$ are the smallest $N$ eigenvalues of the operator $H$, assuming that $\lambda_N<\lambda_{N+1}$. Here, $\av{S^{d-1}}$ is the Lebesgue measure of the unit sphere $S^{d-1}$ ($|S^0|=2$, $|S^1|=2\pi$, $|S^2|=4\pi$). The existence of a solution of~\eqref{eq:SCF_finite} is due to Lieb and Simon~\cite{LiebSimon}. 

In order to describe infinite systems, it is more convenient to reformulate the rHF problem in terms of the \emph{one-particle density matrix} formalism~\cite{stability}. In this formalism, the state of the electrons is described by the orthogonal projector $\gamma=\sum_{i=1}^N\av{\phi_i\rangle \langle \phi_i}$ of rank $N$ and the equations~\eqref{eq:SCF_finite} can be recast as
\begin{equation}\label{eq:SCF_infini}
\left\{ 
\begin{array}{l}
  \gamma=\1\bra{ H\leq \epsilon_F}\\[0,2cm]
\displaystyle H=-\frac12 \Delta +V\\[0,2cm]
\displaystyle
-\Delta V+m^2V=\av{S^{d-1}}\bra{ \rho_\gamma-\nu_{\rm nuc}},
 \end{array}
\right. 
\end{equation}
where formally $\rho_\gamma(x)=\gamma(x,x)$ and the Fermi level $\epsilon_F$ is any real number in the gap $[\lambda_N,\lambda_{N+1})$.

For infinite systems, the rHF equation is still given by~\eqref{eq:SCF_infini}, but $\gamma$ is now an infinite rank operator as there are infinitely many electrons in the system. The operator $\gamma$ needs to be locally trace class for the electronic density $\rho_\gamma$ to be well-defined in $L^1_{ \rm loc}(\RRd)$.

The rHF equation~\eqref{eq:SCF_infini} was solved for periodic nuclear densities 
$$\nu_{\rm nuc}=\nu_\per=\sum_{k\in \mathcal R}\eta(\cdot-k)$$
by Catto, Le Bris and Lions in~\cite{CLL_periodic}, and periodic nuclear densities with local perturbations 
$$\nu_{\rm nuc} =\sum_{k\in \mathcal R}\eta(\cdot-k)+\nu$$ 
were studied by Cancès, Deleurence and Lewin in~\cite{CDL}. We have denoted by $\mathcal R$ the underlying discrete periodic lattice. The corresponding Hamiltonians are denoted by $H_\per$ and $H_\nu$. 
Stochastic distributions,
$$ \nu_{\rm nuc }(\o,\cdot)=\sum_{k\in \mathcal R}\eta(\cdot-k)+\sum_{k\in \mathcal R}q_k(\o)\chi(\cdot-k)$$
for instance, were treated in~\cite{CaLaLe-12}. 

Our present work follows on from~\cite{CDL,CL,CaLaLe-12}. We are going to solve the equation~\eqref{eq:SCF_infini} in the particular case where 
\begin{align}\label{eq:arbitrary_defect}
\nu_{\rm nuc} =\nu_\per+\nu,
\end{align}
where $\nu_\per$ is a periodic nuclear distribution so that the corresponding background crystal is an insulator (the mean-field Hamiltonian $H_\per$ has a gap around $\epsilon_F$), and $\nu\in L^2_\unif(\RRd)$ is a small enough arbitrary perturbation of the background crystal. The perturbation $\nu$ needs to be small in amplitude locally, but must not be local or have any spatial invariance. 

\medskip

The rHF model is an approximation of the $N$-body Schrödinger model, for which there is no well-defined formulation for infinite systems so far. The only available result is the existence of the thermodynamic limit of the energy: the energy per unit volume of the system confined to a box, with suitable boundary conditions, converges when the size of the box grows to infinity. The first theorem of this form for Coulomb interacting systems is due to Lieb and Lebowitz in~\cite{LieLeb-72}. In this latter work, nuclei are considered as quantum particle and rotational invariance plays a crucial role. For quantum systems in which the nuclei are classical particles, the thermodynamic limit was proved for perfect crystals by Fefferman~\cite{Fefferman-85} (a recent proof has been proposed in~\cite{HaiLewSol_2-09}) and for stationary stochastic systems by Blanc and Lewin~\cite{BlaLew-12}. Similar results for Yukawa interacting systems are simpler than for the Coulomb case and follow from the work of Ruelle and Fisher~\cite{RF-66} for perfect crystals and Veniaminov~\cite{Veniaminov} for stationary stochastic systems. Unfortunately, very little is known about the limiting quantum state in both cases. 

For (orbital-free) Thomas-Fermi like theories, the periodic model was studied in~\cite{LiebSimon, CLL_book}, the case of crystals with local defects was studied in~\cite{CanEhr} and stochastic systems were investigated in~\cite{BLBL2007}. To the best of our knowledge, the only works dealing with systems with arbitrary distributed nuclei are~\cite{CLL_book,BLL-03} for Thomas-Fermi type models. 

As mentioned before, our work is the first one to consider this kind of systems in the framework of Hartree-Fock type models. Our results concern small perturbations of perfect crystals interacting through short-range Yukawa potential. Extending these results to more general geometries and for the long-range Coulomb interaction are important questions that we hope to address in the future. 

\medskip

After having found solutions of~\eqref{eq:SCF_infini} for any (small enough) $\nu\in L^2_\unif(\RRd)$, we study
the properties of this solution for local perturbations $\nu$. This enables us to investigate small random perturbations of perfect crystals. Precisely, we consider nuclear distributions 
$$
\nu_{\rm nuc}(\o,x)=\nu_\per(x)+\sum_{k\in \mathcal R}q_k(\o)\chi(x-k),
$$
where $(q_k)_{k\in\mathcal R}$ are i.i.d. Bernoulli variables of parameter $p$ and $\chi$ is a compactly supported function which is small enough in $L^2(\RRd)$. We are interested in the properties of the system in the limit of low concentration of defects, that is when the parameter $p$ goes to zero. We prove that the density of states  of the mean-field Hamiltonian $H_p=-\frac12\Delta+V_p$, which describes the collective behavior of the electrons, admits an expansion of the form
\begin{equation}\label{eq:asymptotic_expansion}
  n_p= n_0+  \sum_{j=1}^J\mu_jp^j+O(p^{J+1}).
\end{equation}
Here, $ n_0$ is the density of states  of the unperturbed Hamiltonian $H_\per=-\frac{1}{2}\Delta+V_\text{per}$ and $\mu_1$ is a function of the spectral shift function for the pair of operators $H_\per$ and $H_\chi$, the latter being the mean-field Hamiltonian of the system with only one local defect constructed in~\cite{CDL}. We give in Theorem~\ref{th:expansion} a precise meaning of $O(p^{J+1})$.

In~\cite{Klopp-95}, Klopp considers the empirical linear Anderson-Bernoulli model 
$$
H=-\frac{1}{2}\Delta+ V_0+V\quad \text{with}\quad V(\o,x)=\sum_{k\in\mathcal R}q_k(\o)\eta(x-k),
$$
where $V_0 $ is a linear periodic potential and $\eta$ an exponentially decaying potential. He proves that the density of states  of the Hamiltonian $H$ admits an asymptotic expansion similar to~\eqref{eq:asymptotic_expansion}. The case where $V(\o,x)$ is distributed following a Poisson law instead of Bernoulli is dealt with in~\cite{Klopp-poisson}. Our proof of~\eqref{eq:asymptotic_expansion} follows the same lines as the one of Klopp. The main difficulty here is to understand the decay properties of the mean-field potential $V$ solution of the self-consistent equations~\eqref{eq:SCF_infini}. For this reason, we dedicate an important part of this paper to the study of these decay properties. In Theorem~\ref{th:decay} below, we show that for a compactly supported perturbation $\nu$, the difference $V-V_\per$ decays faster than any polynomial far from the support of the perturbation $\nu$. Moreover, we show  that the potential generated by two defects that are far enough is close to the sum of the potentials generated by each defect alone. 

\medskip

The article is organized as follow. In Section~\ref{sec:main_results}, we present the main results of the paper. We start by recalling the reduced Hartree-Fock model for perfect crystals and perfect crystals with local defects in Section~\ref{sec:perfect_crystal}. In Section~\ref{sec:res_existence}, we state the existence of solutions to the self-consistent equations~\eqref{eq:SCF_infini} for $\nu_{\rm nuc}$ given by~\eqref{eq:arbitrary_defect}. We also explain that our solution is in some sense the minimizer of the energy of the system. We also prove a thermodynamic limit, namely, the ground state of the system with the perturbation $\nu$ confined to a box converges, when the size of the box goes to infinity, to the ground state of the system with the perturbation $\nu$. In Section~\ref{sec:res_decay}, we prove decay estimates for the mean-field density and potential. In Section~\ref{sec:res_expansion}, we present the expansion of the density of states  of the mean-field Hamiltonian. The proofs of all these results are provided in Sections~\ref{sec:existence_GS},~\ref{sec:decay_rate},~\ref{sec:thermo_limit} and~\ref{sec:expansion}. 
In Section~\ref{sec:linear_response}, we study the dielectric response of a perfect crystal to a variation of the effective charge distribution, which plays a key role in this paper.

\bigskip

\noindent\textbf{Acknowledgement.}
I thoroughly thank \'Eric Cancès and Mathieu Lewin for their precious help and advices. The research leading to these results has received funding from the European Research Council under the European Community's Seventh Framework Programme (FP7/2007--2013 Grant Agreement MNIQS no. 258023).


\section{Statement of the main results}\label{sec:main_results}


\subsection{The rHF model for crystals with and without local defects}\label{sec:perfect_crystal}

In defect-free materials, the nuclei and electrons are arranged according to a discrete periodic lattice $\mathcal R$ of $\RRd$, in the sense that both the nuclear density $\nu_{\rm nuc}=\nu_\per$ and the electronic density are $\mathcal R$-periodic functions. For simplicity, we take $\mathcal R=\ZZd$ in the following. The reduced Hartree-Fock model for perfect crystals has been rigorously derived from the reduced Hartree-Fock model for finite molecular systems by means of thermodynamic limit procedure in~\cite{CLL_periodic,CDL} in the case of Coulomb interaction. The same results for Yukawa interaction are obtained with similar arguments. The self-consistent equation~\eqref{eq:SCF_infini} then reads
\begin{equation}\label{eq:SCF_periodic}
\left\{ 
\begin{array}{l}
  \gamma_0=\1 \bra{H_\per\leq \epsilon_F}\\[0,2cm]
\displaystyle H_\per=-\frac12 \Delta +V_\per\\[0,2cm]
\displaystyle -\Delta V_\per+m^2V_\per=\av{S^{d-1}}\bra{ \rho_{\gamma_0}-\nu_{\rm per}}.
 \end{array}
\right. 
\end{equation}
It has been proved in~\cite{CLL_periodic,CDL} that~\eqref{eq:SCF_periodic} admits a unique solution which is the unique minimizer of the periodic rHF energy functional. 

Most of our results below hold only for insulators (or semi-conductors). We therefore make the assumption that
\begin{equation}\label{eq:Assumption_H_0_has_a_gap}
H_\per\text{ has a spectral gap around } \epsilon_F.
\end{equation}

\medskip

The rHF model for crystals with local defects was introduced and studied in~\cite{CDL}. A solution of the rHF equation~\eqref{eq:SCF_infini} is constructed using a variational method. One advantage of this method is that there is no need to assume that the perturbation $\nu$ is small in amplitude. The idea is to find a minimizer of the infinite energy of the system by minimizing the energy difference between the perturbed state and the perfect crystal. The ground state density matrix can thus be decomposed as
\begin{align}\label{eq:probleme_defaut_local}
 \gamma=\gamma_0+Q_\nu,
\end{align}
where $Q_\nu$ is a minimizer of the energy functional 
\begin{equation}\label{eq:energie_defaut} 
\E^\nu(Q)=\tr_{\gamma_0}\bra{(H_\per-\epsilon_F) Q}+\frac{1}{2}D_m(\rho_Q-\nu,\rho_Q-\nu)
\end{equation} 
on the convex set
\begin{equation}\label{eq:K}
\begin{array}{c}
 \K=\left\{Q^*=Q,\; -\gamma_0\leq Q\leq 1-\gamma_0,\;\bra{-\Delta+1}^\frac12 Q\in\S_2(L^2(\RRd)),\right.\\
\left.\;\bra{-\Delta+1}^\frac12 Q^{\pm\pm} \bra{-\Delta+1}^\frac12 \in \S_1(L^2(\RRd)) \right\},
\end{array}
\end{equation}
where $Q^{++}=(1-\gamma_0)Q(1-\gamma_0)$, $Q^{--}=\gamma_0 Q \gamma_0$ and $\tr_{\gamma_0}(A)=\tr\bra{A^{++}+A^{--}}$. We use the notation $\S_p$ to denote the $p^{\rm th}$ Schatten class. In particular $\S_2$ is the set of Hilbert-Schmidt operators. 
The second term of~\eqref{eq:energie_defaut} accounts for the interaction energy and is defined for any charge densities $f,g\in H^{-1}(\RRd)$  by
\begin{align*}
D_m(f,g)&=\av{S^{d-1}}\int_\RRd\frac{\overline{\widehat{f}(p)}\widehat{g}(p)}{\av{p}^2+m^2}dp=\int_{\RRd}\int_{\RRd}f(x)Y_m(x-y)g(y)\,dx\,dy,
\end{align*}
where $\widehat{f}(p)=\bra{2\pi}^{-\frac{d}{2}}\int_\RRd f\bra{x}e^{-ip\cdot x}dx$ is the Fourier transform of $f$.
The Yukawa kernel $Y_m$, the inverse Fourier transform of $\av{S^{d-1}}(\av{p}^2+m^2)^{-1}$, is given by
$$
Y_m(x)=\left\{
\begin{array}{ll}
m^{-1}e^{-m\left|x \right|}\; & \mbox{if}\;d=1,\\ [0,2cm]
\displaystyle
K_0\bra{m\av{x}}  \;  &\mbox{if}\;d=2,\\ [0,2cm]
\displaystyle
|x|^{-1}e^{-m\left|x \right|} \; &\mbox{if}\;d=3,
\end{array}
\right.
$$
where $K_0\bra{r}=\int_0^\ii e^{-r\cosh t}\,dt$ is the modified Bessel function of the second type~\cite{LL}. It has been proved in~\cite{CDL} that the energy functional~\eqref{eq:energie_defaut} is convex and that all its minimizers share the same density $\rho_\gamma$. These minimizers are of the form
\begin{equation}\label{eq:SCF_defaut}
 \left\{
\begin{array}{l}
\gamma=\1 \bra{H\leq \epsilon_F}+\delta\\[0,2cm]
\displaystyle
 H=-\frac12 \Delta+ V\\[0,2cm]
\displaystyle
-\Delta V+ m^2V=\av{S^{d-1}} (\rho_\gamma-\nu_\per-\nu),
\end{array}
\right.
\end{equation}
where $0 \leq \delta\leq \1\bra{H=\epsilon_F}$. If $\nu$ is small enough in the $H^{-1}$-norm, then $\delta=0$.

One of the purposes of this article is to find decay estimates of the potential $V$ solution of~\eqref{eq:SCF_defaut} that are necessary in the study of the Anderson-Bernoulli random perturbations of crystals.


\subsection{Existence of ground states}\label{sec:res_existence}

In this section, we state our results concerning the electronic state of a perturbed crystal. The host crystal is characterized by a periodic nuclear density $\nu_\per\in L^2_\unif(\RRd)$ such that the gap assumption~\eqref{eq:Assumption_H_0_has_a_gap} holds. The perturbation is given by a distribution $\nu\in L^2_\unif(\RRd)$. The total nuclear distribution is then 
$$
\nu_{\rm nuc} =\nu_\per+\nu.
$$
In Theorem~\ref{th:point_fixe} below, we show that if $\nu$ is small enough in the $L^2_\unif$-norm, then the rHF equation~\eqref{eq:SCF_infini} 
admits a solution $\gamma$. This solution is unique in a neighborhood of $\gamma_0$. The proof consists in formulating the problem in terms of the density $\rho_\gamma$ and using a fixed point technique, in the spirit of~\cite{HLS-05}.

\begin{theorem}[Existence of a ground state]\label{th:point_fixe}
There exists $\alpha_c>0$ and $C\geq 0$ such that for any $\nu\in L^2_\unif(\RRd)$ satisfying $\norm{\nu}_{L^2_\unif}\leq \alpha_c$, 
there is a unique solution $\gamma\in \S_{1,\text{\rm loc}}(L^2(\RRd))$ to the self-consistent equation 
\begin{equation}\label{eq:SCE}
 \left\{
\begin{array}{l}
\gamma=\1 \bra{H\leq \epsilon_F}\\[0,2cm]
\displaystyle
 H=-\frac12\Delta+V\\[0,2cm]
\displaystyle
-\Delta V+ m^2 V=\av{S^{d-1}} (\rho_\gamma-\nu-\nu_\per)
\end{array}
\right.
\end{equation}
satisfying
\begin{equation}\label{eq:rho_controle_par_nu}
 \norm{\rho_\gamma-\rho_{\gamma_0}}_{L^2_\unif}\leq C\norm{\nu}_{L^2_\unif}.
\end{equation}
We denote this solution by $\gamma_\nu$, the response electronic density by $\rho_{\nu}=\rho_{\gamma_\nu}-\rho_{\gamma_0}$ and the defect mean-field potential by $V_\nu=V-V_\per$.
\end{theorem}

For a local defect $\nu\in L^2(\RRd)\cap L^1(\RRd)$ such that $\norm{\nu}_{L^2_\unif}\leq \alpha_c$, equation~\eqref{eq:SCE} admits a unique solution which coincides with the ground state $\gamma$ solution of~\eqref{eq:probleme_defaut_local} constructed in~\cite{CDL}. 
Indeed, the solution $\gamma_\nu$ given in Theorem~\ref{th:point_fixe} is a solution of the defect problem~\eqref{eq:SCF_defaut}. Moreover, in the proof of Theorem~\ref{th:point_fixe}, we prove that $H$ has a gap around $\epsilon_F$, thus necessarily $\delta=0$ in~\eqref{eq:SCF_defaut}. As all the solutions of~\eqref{eq:SCF_defaut} share the same density,~\eqref{eq:SCF_defaut} (thus~\eqref{eq:SCE}) admits a unique solution. 

The ground state constructed in Theorem~\ref{th:point_fixe} is in fact the unique minimizer of the "infinite" rHF energy functional. Indeed, following ideas of~\cite{HLS-07}, we can define the relative energy of the system with nuclear distribution $ \nu_{\rm nuc} $ by subtracting the "infinite" energy of $\gamma_\nu$ from the "infinite" energy of a test state $\gamma$:
\begin{align*}
\displaystyle \E_{ \nu}^{\rm rel}(\gamma):=\tr_{\gamma_\nu}\bra{\bra{H-\epsilon_F}\bra{\gamma-\gamma_\nu}}+\frac12 D_m\bra{ \rho_\gamma-\rho_{\gamma_\nu},\rho_\gamma-\rho_{\gamma_\nu} }.
\end{align*}
This energy is well-defined for states $\gamma$ such that $\gamma-\gamma_\nu$ is finite rank and smooth enough for instance, but one can extend it to states in a set similar to $\K$ in~\eqref{eq:K}. The minimum of the energy $\E_{ \nu}^{\rm rel}$ is attained for $\gamma= \gamma_\nu=\1 \bra{H\leq \epsilon_F}$. Moreover, as $H$ has a gap around $\epsilon_F$, $\E_{ \nu}^{\rm rel}$ is strictly convex and $\gamma_\nu$ is its unique minimizer. 

In the following theorem, we show that if we confine the defect $\nu$ to a box of finite size, then the ground state of the system defined by the theory of local defects presented in Section~\ref{sec:perfect_crystal} converges, when the size of the box goes to infinity, to the ground state of the system with the defect $\nu$ defined in Theorem~\ref{th:point_fixe}. We denote by $\Gamma_L=\left[-L/2,L/2\right)^d$.

\begin{theorem}[Thermodynamic limit]\label{th:limite_thermo}
There exists $\alpha_c> 0$ such that for any $\nu\in L^2_\unif(\RRd)$ satisfying $\norm{\nu}_{L^2_\unif}\leq \alpha_c$, the sequence $(\gamma_{\nu\1_{{\Gamma}_L}})_{L\in \NN\setminus\set{0}}$ converges in $\S_{1,\text{loc}}(L^2(\RRd))$ to $\gamma_\nu$ as $\cL$.
\end{theorem}


\subsection{Decay estimates }\label{sec:res_decay}
In this section, we prove some decay estimates of the mean-field potential $V_\nu$ and the mean-field density $\rho_\nu$, which will be particularly important to understand the system in the presence of rare perturbations in the next section.

Theorem~\ref{th:decay} below is crucial in the proof of Theorem~\ref{th:expansion}. Indeed, we will need uniform decay estimates for compactly supported defects, with growing supports and uniform local norms. 

\begin{theorem}[Decay rate of the mean-field potential and density]\label{th:decay}
There exists $\alpha_c,C'>0 $ and $C\geq 0$ such that for any $\nu\in L^2_c(\RRd)$ satisfying $\norm{\nu}_{L^2_\unif}\leq \alpha_c$, we have for $R\geq 2$
\begin{equation}\label{eq:decayL2unif}
\norm{V_\nu}_{H^2_\unif(\RRd\setminus C_R(\nu))}+\norm{\rho_\nu}_{L^2_\unif(\RRd\setminus C_R(\nu))}\leq C e^{-C'\bra{\log R}^2}\norm{\nu}_{L^2_\unif(\RRd)},
\end{equation}
where
$C_R(\nu)=\set{x\in \RRd, \; \text{\rm d}\bra{x, \text{\rm supp}(\nu)}< R}$. 
\end{theorem}

\begin{remark}
Using the same techniques as in the proof of Theorem~\ref{th:decay}, we can prove (see~\cite{these}) that there exists $\alpha,\alpha_c,C'>0$ and $C\geq 0$ such that for any $\nu\in L^2_c(\RRd)$ satisfying $\norm{\nu}_{L^2_\unif}\leq \alpha_c$ and $\norm{\nu}_{H^{-1}}\leq \alpha$, we have for $R\geq 2$ 
\begin{equation}\label{eq:decayL2}
\norm{V_\nu}_{H^2(\RRd\setminus C_R(\nu))}+\norm{\rho_\nu}_{L^2(\RRd\setminus C_R(\nu))}\leq C e^{-C'\bra{\log R}^2}\norm{\nu}_{L^{2}(\RRd)}.
\end{equation}
Estimate~\eqref{eq:decayL2} gives a decay rate of the solution of the rHF equation for crystals with local defects, far from the support of the defect. In particular, it shows that $\rho_\nu\in L^1(\RRd)$. This decay is due to the short-range character of the Yukawa interaction. In the Coulomb case, it has been proved in~\cite{CL} that for anisotropic materials, $\rho_\nu\notin L^1(\RRd)$.
\end{remark}

The decay rate of $V_\nu$ and $\rho_\nu$ proved in Theorem~\ref{th:decay} is faster than the decay of any polynomial, but is not exponential, which we think should be the optimal rate.

Proposition~\ref{prop:loc} below is an important intermediary result in the proof of Theorem~\ref{th:limite_thermo}.  It says that the mean-field density $\rho_\nu$ and potential $ V_\nu$ on a compact set depend mainly on the nuclear distribution in a neighborhood of this compact set.

\begin{proposition}[The mean-field potential and density depend locally on $\nu$]\label{prop:loc}
There exists $\alpha_c>0 $ such that for any $\beta\geq 2$ there exists $C\geq 0$ such that for any $\nu\in L^2_\unif(\RRd)$ satisfying $\norm{\nu}_{L^2_\unif}\leq \alpha_c$ and any $L\geq 1$ we have 
$$
\norm{V_\nu -V_{\nu_L}}_{H^2_\unif(B(0, L/4^\beta))}+ \norm{\rho_\nu -\rho_{\nu_L}}_{L^2_\unif(B(0, L/4^\beta))}\leq \frac{C}{L^\beta}\norm{\nu}_{L^2_\unif},
$$
where $\nu_L=\nu\1_{{\Gamma} _L}$. 
\end{proposition}

In the same way, we obtain the following result which will be very useful in the proof of Theorem~\ref{th:expansion}. We prove that the potential generated by two defects that are far enough is close to the sum of the potentials generated by each defect alone in the sense of

\begin{proposition}\label{prop:cor:loc}
There exists $\alpha_c>0 $ such that for any $\beta\geq 2$, there exists $C\geq 0$ such that for any $\nu_1,\nu_2\in L^2_c(\RRd)$ satisfying $\norm{\nu_1}_{L^2_\unif},\norm{\nu_2}_{L^2_\unif}\leq \alpha_c$ and $R={\rm d}( {\rm supp}(\nu_1),{\rm supp}(\nu_2))> 0$, we have
\begin{align*}
&\norm{V_{\nu_1+\nu_2} -V_{\nu_2}}_{H^2_\unif( C_{R/4^\beta}(\nu_2))}+ \norm{\rho_{\nu_1+\nu_2} -\rho_{\nu_2}}_{L^2_\unif( C_{R/4^\beta}(\nu_2))}
\\
&\qquad\qquad\qquad\qquad\qquad\qquad\qquad\qquad\qquad\qquad
\leq \frac{C}{R^\beta}\bra{\norm{\nu_1}_{L^2_\unif}+\norm{\nu_2}_{L^2_\unif}}.
\end{align*}
\end{proposition}

\begin{proof}
The proof is the same as the one of Proposition~\ref{prop:loc} with $\nu=\nu_1+\nu_2$ and $L=2R$.
\end{proof}

\subsection{Asymptotic expansion of the density of states}\label{sec:res_expansion}

In this section, we use our previous results to study a particular case of random materials. 
In the so-called statistically homogeneous materials, the particles are randomly distributed over the space with a certain spatial invariance. More precisely, the nuclear distribution (thus the electronic density) is stationary in the sense
$$
\nu_{\rm nuc}(\tau_k(\o),x)=\nu_{\rm nuc}(\o,x+k),
$$ 
where $(\tau_k)_{k\in \ZZd}$ is an ergodic group action of $\ZZd$ on the probability set $\O$ (see Figure~\ref{fig:mu_stat}). 
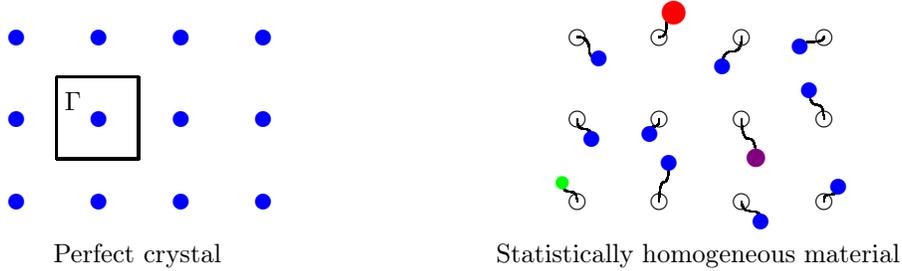
\begin{figure}[h]
\centerline{
 \begin{tabular}{ccc}
 & & \\
 \setlength{\unitlength}{1200sp}%

\begingroup\makeatletter\ifx\SetFigFont\undefined%
\gdef\SetFigFont#1#2#3#4#5{%
  \reset@font\fontsize{#1}{#2pt}%
  \fontfamily{#3}\fontseries{#4}\fontshape{#5}%
  \selectfont}%
\fi\endgroup%

\begin{picture}(5000,5000)(0,-4000)
\put(825,-825){\line( 1,0){1700}}
\put(2525,-825){\line( 0,-1){1700}}
\put(2525,-2525){\line( -1,0){1700}}
\put(825,-2525){\line( 0,1){1700}}
\put(1000,-1500){{$\Gamma$}}
{\color[rgb]{0,0,1}\thinlines\put(  0,0){\circle*{350}}}%
{\color[rgb]{0,0,1} \put(1700,0){\circle*{350}}}%
{\color[rgb]{0,0,1}\put(3400,0){\circle*{350}} }%
{\color[rgb]{0,0,1}\put(5100,0){\circle*{350}} }%
{\color[rgb]{0,0,1}\put(  0,-1700){\circle*{350}} }%
{\color[rgb]{0,0,1}\put(1700,-1700){\circle*{350}} }%
{\color[rgb]{0,0,1}\put(3400,-1700){\circle*{350}} }%
{\color[rgb]{0,0,1}\put(5100,-1700){\circle*{350}}  }%
{\color[rgb]{0,0,1}\put(0,-3400){\circle*{350}}}%
{\color[rgb]{0,0,1}\put(1700,-3400){\circle*{350}}}%
{\color[rgb]{0,0,1}\put(3400,-3400){\circle*{350}}}%
{\color[rgb]{0,0,1}\put(5100,-3400){\circle*{350}}}%

\end{picture}%
 & \quad \quad \quad \quad \quad \quad \quad \quad &
 \setlength{\unitlength}{1200sp}%

\begingroup\makeatletter\ifx\SetFigFont\undefined%
\gdef\SetFigFont#1#2#3#4#5{%
  \reset@font\fontsize{#1}{#2pt}%
  \fontfamily{#3}\fontseries{#4}\fontshape{#5}%
  \selectfont}%
\fi\endgroup%

\begin{picture}(5000,5000)(0,-4000)
{\color[rgb]{0,0,0}\thinlines\put(  0,0){\circle{350}}}%
{\color[rgb]{0,0,0} \put(1700,0){\circle{350}}}%
{\color[rgb]{0,0,0}\put(3400,0){\circle{350}} }%
{\color[rgb]{0,0,0}\put(5100,0){\circle{350}} }%
{\color[rgb]{0,0,0}\put(  0,-1700){\circle{350}} }%
{\color[rgb]{0,0,0}\put(1700,-1700){\circle{350}} }%
{\color[rgb]{0,0,0}\put(3400,-1700){\circle{350}} }%
{\color[rgb]{0,0,0}\put(5100,-1700){\circle{350}}  }%
{\color[rgb]{0,0,0}\put(0,-3400){\circle{350}}}%
{\color[rgb]{0,0,0}\put(1700,-3400){\circle{350}}}%
{\color[rgb]{0,0,0}\put(3400,-3400){\circle{350}}}%
{\color[rgb]{0,0,0}\put(5100,-3400){\circle{350}}}%

\thinlines
\qbezier(0,0)(250,0)(250, -250)
\qbezier(250, -250)(250,-450)(450, -450)
\qbezier(1700,0)(1850,0)(1850,250)
\qbezier(1850,250)(1850,500)(2000,500)
\qbezier(3400,0)(3400,-300)(3200, -300)
\qbezier(3200,-300)(3000,-300)(3000,-600)
\qbezier(5100,0)(5100,-100)(4850, -100)
\qbezier(4850, -100)(4600,-100)(4600,-200)
\qbezier(  0,-1700)(0,-1900)(150,-1900)
\qbezier(150,-1900)(300,-1900)(300,-2100)
\qbezier(1700,-1700)(1700,-1850)(1600, -1850)
\qbezier(1600, -1850)(1500,-1850)(1500,-2000)
\qbezier(3400,-1700)(3400,-2100)(3550, -2100)
\qbezier(3550, -2100)(3700,-2100)(3700,-2500)
\qbezier(5100,-1700)(5100,-1400)(4950, -1400)
\qbezier(4950, -1400)(4800,-1400)(4800,-1100)
\qbezier(  0,-3400)(0,-3200)(-150, -3200)
\qbezier(-150, -3200)(-300,-3200)(-300,-3000)
\qbezier(1700,-3400)(1700,-3000)(1800, -3000)
\qbezier(1800, -3000)(1900,-3000)(1900,-2600)
\qbezier(3400,-3400)(3400,-3600)(3600, -3600)
\qbezier(3600, -3600)(3800,-3600)(3800,-3800)
\qbezier(5100,-3400)(5100,-3250)(5250, -3250)
\qbezier(5250, -3250)(5400,-3250)(5400,-3100)

{\color[rgb]{0,0,1}\put( 450,-450){\circle*{350}}}%
{\color[rgb]{1,0,0} \put(2000,500){\circle*{500}}}%
{\color[rgb]{0,0,1}\put(3000,-600){\circle*{350}} }%
{\color[rgb]{0,0,1}\put(4600,-200){\circle*{350}} }%
{\color[rgb]{0,0,1}\put(  300,-2100){\circle*{350}} }%
{\color[rgb]{0,0,1}\put(1500,-2000){\circle*{350}} }%
{\color[rgb]{0.5,0,0.5}\put(3700,-2500){\circle*{400}} }%
{\color[rgb]{0,0,1}\put(4800,-1100){\circle*{350}}  }%
{\color[rgb]{0,1,0}\put(-300,-3000){\circle*{250}}}%
{\color[rgb]{0,0,1}\put(1900,-2600){\circle*{350}}}%
{\color[rgb]{0,0,1}\put(3800,-3800){\circle*{350}}}%
{\color[rgb]{0,0,1}\put(5400,-3100){\circle*{350}}}%

\end{picture}%
 \\
Perfect crystal& & Statistically homogeneous material\\
 & & \\
 \end{tabular}
 }
\caption{Example of a stationary nuclear distribution}
\label{fig:mu_stat}
\end{figure} 
One famous example of such distributions is the Anderson model
$$
\nu_{\rm nuc}(\omega,x) = \sum_{k \in \ZZd} q_k(\omega)\, \chi(x-k),
$$
where, typically, $\chi \in C^\infty_c(\RR^3)$ and the $q_k$'s are
{i.i.d. random} variables. The reduced Hartree-Fock model for statistically homogeneous materials was introduced in~\cite{CaLaLe-12}. The state of the electrons is described by a random self-adjoint operator $\bra{\gamma(\o)}_{\o\in\O}$ acting on $L^2(\RRd)$ such that $0\leq \gamma(\o)\leq 1$ almost surely. The rHF equation is then 
\begin{align}\label{eq:SCF_sto}
\left\{
 \begin{array}{l}
\gamma(\o)=\1\bra{H(\o)\leq \epsilon_F}+\delta(\o)\\ [0,2cm]  
\displaystyle H( \o)=-\frac12 \Delta + V(\o,\cdot)\\ [0,2cm]  
\displaystyle -\Delta V(\o,\cdot) +m^2 V(\o,\cdot)=\av{S^{d-1}}\bra{ \rho_{\gamma(\o)}-\nu(\o,\cdot)}
 \end{array}
\right. \quad \text{almost surely,}
\end{align}
where $0\leq \delta (\o)\leq \1_{\set{\epsilon_F}}(H(\o))$ almost surely. The solutions of~\eqref{eq:SCF_sto} turn out to be the minimizers of the energy functional 
\begin{align*}
\underline{ \E}_{\nu_{\rm nuc}}(\gamma)=\tv{\bra{-\frac12\Delta-\epsilon_F} \gamma}+\underline{D}_m(\rho_\gamma-\nu_{\rm nuc},\rho_\gamma-\nu_{\rm nuc}),
\end{align*}
where $\tv{A}=\EE\bra{\tr\bra{\1_{{\Gamma} }  A\1_{{\Gamma} } }}$ and 
$$
\underline{D}_m(f,g)=\EE\bra{\int_\RRd\int_{{\Gamma} }  f(x)Y_m(x-y)g(y)\,dx\,dy}.
$$
Here, ${\Gamma}=\left[-1/2,1/2\right)^d $ denotes the semi-open unit cube. Thanks to the convexity of $\underline{ \E}_{\nu_{\rm nuc}}$, it has been proved in~\cite{CaLaLe-12} that the minimizers of $\underline{ \E}_{\nu_{\rm nuc}}$ share the same density. Therefore, the Hamiltonian $H$ solution of~\eqref{eq:SCF_sto} is uniquely defined.

In this paper, we are interested in the particular case of random perturbation of perfect crystals 
$$
\nu_{\rm nuc}(\omega,x)=\nu_\per(x)+\nu_p(\o,x)
$$
in the limit of low concentration of defects. We restrict our study to Anderson-Bernoulli type perturbations, that is, we suppose that at each site of $\ZZd$, there is a probability $p$ to see a local defect $\chi$, independently of what is happening in the other sites. More precisely, we consider the probability space $\Omega=\{0,1\}^{\mathbb{Z}^d}$ endowed with the measure  $\mathbb P =\bra{p\delta_{1}+ (1-p)\delta_{0} }^{\otimes\ZZd}$ and the ergodic group action  $\tau_k(\omega)=\omega_{\cdot+k}$. The defect distribution we consider is then given by
$$
\nu_p(\o,x)=\sum_{k\in\ZZd}q_k(\o)\chi(x-k)
$$
where $q_k$ is the $k^{\rm th}$ coordinates of $\o$ and $\chi\in L^2(\RRd )$ with ${\rm supp}(\chi)\subset\Gamma$. The $q_k$'s are i.i.d. Bernoulli variables of parameter $p$. If $\norm{\chi}_{L^2} \leq \alpha_c$, then $\delta(\o)=0$ almost surely and~\eqref{eq:SCF_sto} admits a unique solution. For almost every $\o$, this solution coincides with the solution of~\eqref{eq:SCE} constructed in Theorem~\ref{th:point_fixe}. For convenience, we will from now on use the notation 
$$
H_0=H_\per-\epsilon_F,
$$
where we recall that $\epsilon_F$ is the Fermi level. We introduce the mean-field Hamiltonian corresponding to the system with the defect $\nu_p$ 
$$
H_p=H_0+V_{\nu_p}\quad\mbox{with}\quad V_{\nu_p}(\o,x)=Y_m*\bra{\rho_{\nu_p}-\nu_p}.
$$
As $V_p$ is stationary with respect to the ergodic group $(\tau_k)_{k\in\ZZd}$ and uniformly bounded in $\O\times \RRd$, then by~\cite[Theorem 5.20]{pastur}, there exists a deterministic positive measure $ n_p(dx)$, the density of states  of $H_p$, such that for any $\phi$ in the Schwartz space $\mathcal S(\RR)$
$$
\int_\RR\phi(x) n_p(dx)=\tv{\phi(H_p)}.
$$
For $K \subset \ZZd$, we define the self-consistent operator corresponding to the system with the defects in $K$
$$H_K=H_0+ V_K,$$ 
where
$$
V_K=Y_m*(\rho_K-\nu_K),\quad\text{}\quad \nu_K=\sum_{k\in K}\chi(\cdot-k)\quad \text{and} \quad\rho_K=\rho_{\nu_K}. 
$$
If $\av{K}<\ii$, we denote by $\xi_K(x)$ the spectral shift function~\cite{Yafaev} for the pair of operators $H_K$ and $H_0$. It is the tempered distribution in $\mathcal S'(\RR)$ satisfying, for any $\phi\in\mathcal S(\RR)$,
$$
\tr\bra{\phi(H_K)-\phi(H_0)}=\int_\RR \xi_K(x)\phi'(x)\,dx =-\int_\RR \xi'_K(x)\phi(x)\,dx. 
$$

In Theorem~\eqref{th:expansion} below, we give the asymptotic expansion of the density of states $ n_p$ in terms of powers of the Bernoulli parameter $p$. 
\begin{theorem}\label{th:expansion} 
For $\chi\in L^2(\RRd)$ such that ${\rm supp}(\chi)\subset \Gamma$ and $K\subset \ZZd$ such that $\av{K}<\ii$, we define the tempered distribution $\mu_K$ by
$$
\mu_K(x)=-\frac{1}{\av{K}}\sum_{K'\subset K}(-1)^{\av{K\setminus K'}}\xi_{K'}'(x).
$$
There exists $\alpha_c>0$ such that if $\norm{\chi}_{L^2}\leq \alpha_c$, then
\renewcommand{\labelenumi}{(\roman{enumi})}
\begin{enumerate}
 \item for $j\in \set{1,2}$, $\mu_j=\sum_{K\subset\ZZd,\atop \av{K}=j,\, 0\in K} \mu_K$ is a well-defined convergent series in $\mathcal S'(\RR)$.
\item for $J\leq 2$, there exists $C_J\geq 0$, independent of $\chi$ such that for any $\phi\in \mathcal S(\RR)$,
$$
\av{\langle  n_p, \phi \rangle- \langle  n_0, \phi \rangle -\sum_{j=1}^J \langle \mu_j, \phi \rangle}\leq C_J\norm{\chi}_{L^2} \sum_{\alpha\leq (J+3)(d+1) \atop \beta \leq J+4 +(J+2)d}\mathcal N_{\alpha,\beta}\bra{\phi}p^{J+1},
$$
where $ n_0$ is the density of states  of the unperturbed Hamiltonian $H_0$ and $\mathcal N_{\alpha,\beta}(\phi)=\sup_{x\in \RR}\av{x^\alpha\frac{\partial ^\beta \phi}{\partial x^\beta}}$.
\end{enumerate} 
\end{theorem}
In Theorem~\ref{th:expansion}, we only present the expansion of the density of states until the second order $J=2$. The proof of the expansion up to any order $J\in \NN$ should follow the same lines and techniques used here. 

A result similar to Theorem~\ref{th:expansion} was obtained in~\cite{Klopp-95} in the linear case. Materials with low concentration of defects were studied by Le Bris and Anantharaman~\cite{AnLe-2}.
in the framework of stochastic homogenization. 

The proof of Theorem~\ref{th:expansion} follows essentially the proof of~\cite[Theorem 1.1]{Klopp-95}. It uses the decay of the potential related to each local defect. In~\cite[Theorem 1.1]{Klopp-95}, the linear potential is assumed to decay exponentially. In our nonlinear model, the decay estimates established in Section~\ref{sec:res_decay} play a crucial role in the proof. 

\medskip

The rest of the paper is devoted to the proofs of the results presented in this section. In the next section, we study the dielectric response of the crystal to an effective charge perturbation. The results of Section~\ref{sec:linear_response} will be used in later sections.


\section{Dielectric response for Yukawa interaction}\label{sec:linear_response}

In this section, we study the dielectric response of the electronic ground state of a crystal to a small effective charge perturbation $f\in L^2_\unif(\RRd)$. This means more precisely that we expand the formula
$$
Q_f=\1\bra{H_0+f*Y_m\leq 0}-\1\bra{H_0\leq 0}
$$
in powers of $f$ (for $f$ small enough) and state important properties of the first order term. The higher order term will be dealt with later in Lemma~\ref{lemma:LetL2continus}. For Coulomb interactions and local perturbation $f\in L^2(\RRd)\cap \C_0(\RRd)$, where $\C_0(\RRd)$ is the Coulomb space, this study has been carried out in~\cite{CL} in dimension $d=3$. 

The results of this section can be used in the linear model or the mean-field framework. 
In the reduced Hartree-Fock model we consider in this paper, the effective charge perturbation is $f=\rho_{\nu}-\nu$, where $\rho_\nu$ is the electronic density of the response of the crystal to the nuclear perturbation $\nu$ defined in Theorem~\ref{th:point_fixe}. Expanding (formally) $Q_f$ in powers of $f$ and using the resolvent formula leads to considering the following operator 
$$
Q_{1,f}=\frac{1}{2i\pi}\oint_\C\frac{1}{z-H_0}f*Y_m \frac{1}{z-H_0}dz,
$$
where $\C$ is a smooth curve in the complex plane enclosing the whole spectrum of $H_0$ below $0$ (see Figure~\ref{fig:gap}). 
\begin{figure}[!h]
\psfrag{sigma}{$\sigma(H_0)$} \psfrag{0}{$0$} \psfrag{C}{$\C$}
\includegraphics[scale=0.5]{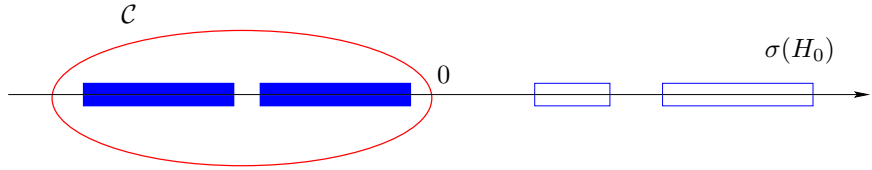}
\caption{Graphical representation of a contour $\C\subset \CC$ enclosing $\sigma(H_0)\cap (-\ii, 0]$.}
\label{fig:gap}
\end{figure}
By the residue Theorem, the operator $Q_{1,f}$ does not depend on the particular curve $\C$ chosen as above. We recall that $V_\per$ is $-\Delta$ bounded with relative bound $0$. Thus $H_0$ is bounded below by the Rellich-Kato theorem~\cite[Theorem X.12]{ReeSim2}. Theorem~\ref{th:props_de_L} below studies the properties of the dielectric response operator $\L:f\rightarrow \rho_{Q_{1,f}}$ and the operator $\bra{1+\L}^{-1}$, which will play an important role in the resolution of the self-consistent equation~\eqref{eq:SCE}. In particular, it gives the functional spaces on which $\L$ and $\bra{
1+\L}^{-1}$ are well-defined for both local and extended charge densities. It also says that $\bra{
1+\L}^{-1}$ is local in the sense that its off-diagonal components decay faster than any polynomial.
We consider $H^{-1}(\RRd)$, endowed with the scalar product
$$
\langle f, g\rangle_{H^{-1}} = \frac{1}{(2\pi)^d} \int_\RRd \frac{\overline{\widehat{f}(p)}\widehat{g}(p)}{\av{p}^2+m^2}dp.
$$

\begin{theorem}[Properties of the dielectric response]\label{th:props_de_L}
We have
\renewcommand{\labelenumi}{(\roman{enumi})}
\begin{enumerate}
 \item 
The operator
$$
\begin{array}{lrll}
 \L: & H^{-1}(\RRd)&\rightarrow & H^{-1}(\RRd) \\
 & f&\mapsto &-\rho_{Q_{1,f}},
\end{array}
$$ 
is well-defined, bounded, non-negative and self-adjoint. Hence $1+\L$ is invertible and bicontinuous. 

\item The operator $\L$ is bounded from $H^{-1}(\RRd)$ to $L^2(\RRd)$ and $1/(1+\L)$ is a well-defined, bounded operator from $L^2(\RRd)$ into itself. 
\item The operator
$$
\begin{array}{lrll}
 \L: & L^2_\unif(\RRd)&\rightarrow & L^2_\unif(\RRd)\\
 & f&\mapsto &-\rho_{Q_{1,f}},
\end{array}
$$ 
is well-defined and bounded. The operator $1+\L$ is invertible on $L^2_\unif(\RRd)$ and its inverse is bounded. 

\item There exist $C\geq 0$ and $C'>0$ such that for any $j,k\in \ZZd$ such that $\av{k-j}\geq 1$, we have
\begin{equation}\label{eq:1/1+L_localise_unif} 
\norm{\1_{\Gamma+j}\frac{1}{1+\L}\1_{\Gamma+k} }_{\B}\leq Ce^{-C'\bra{\log \av{k-j}}^2}.
\end{equation}
\end{enumerate}
\end{theorem}

\begin{proof}

The proof consists in the following 6 steps. In the whole paper $C\geq 0$ and $C'>0$ are constants whose value might change from one line to the other.
\paragraph*{Step 1}\textit{Proof of (i).}
The proof is similar to the one of~\cite[Proposition 2]{CL}, with the Yukawa kernel $Y_m$, instead of the Coulomb kernel. 
In the Yukawa case, $H^{-1}(\RRd)$ plays the role of the Coulomb space. The proof of~\cite[Proposition 2]{CL} can easily be adapted to our case. We skip the details for the sake of brevity. 

\paragraph*{Step 2}\textit{Proof of (ii).}
Let $f\in H^{-1}(\RRd)$. Then $Y_m*f\in L^2(\RRd)$ and
\begin{equation}\label{eq:Y_mfinL2}
\norm{Y_m*f}_{L^2}^2=\av{S^{d-1}}^2\int_\RRd\frac{\av{\widehat{f}(p)}^2}{\bra{\av{p}^2+m^2}^2}dp\leq C\int_\RRd\frac{\av{\widehat{f}(p)}^2}{\av{p}^2+m^2}dp= C \norm{f}_{H^{-1}}^2.
\end{equation} 
Therefore, by~\cite[Proposition 1]{CL}, $Q_{1,f}\in \mathcal K$, where $\K$ has been defined in~\eqref{eq:K}, and  $\L f=-\rho_{Q_{1,f}}\in L^2(\RRd)$. Arguing by duality, we have for any $W\in L^2(\RRd)$,
\begin{equation}\label{eq:def_rho}
\tr\bra{Q_{1,f}W}=\int_\RRd\rho_{Q_{1,f}}W.
\end{equation}
Besides,
by the Kato-Seiler-Simon inequality~\cite[Theorem 4.1]{trace_ideals} for $d\leq 3$
\begin{equation}\label{eq:KSS}
 \forall p\geq 2,\quad \norm{f(-i\nabla)g(x)}_{\S_2}\leq (2\pi)^{-\frac{d}{p}}\norm{f}_{L^p}\norm{g}_{L^p}
\end{equation}
 and the fact that 
\begin{equation}\label{eq:B(z)uniformlybounded}
\bra{z-H_0}^{-1}\bra{1-\Delta}\text{ is uniformly bounded on the contour } \C,     
\end{equation}
we have 
$$
\frac{1}{z-H_0}Y_m*f\frac{1}{z-H_0}W\in\S_2(L^2(\RRd))
$$
and 
\begin{align}\label{eq:borneQ1fW}
\av{\tr\bra{Q_{1,f}W}}=\av{ \frac{1}{2i\pi}\oint_\C \tr\bra{\frac{1}{z-H_0}Y_m*f\frac{1}{z-H_0}W}\,dz}&\leq C\norm{Y_m*f}_{L^2}\norm{W}_{L^2}.
\end{align}
The bound~\eqref{eq:B(z)uniformlybounded} follows from the following lemma.

\begin{lemma}\label{lemma:B(z)borné}
 Let $W\in L_\unif^2(\RRd)$. Then there exists $C\geq 0$, depending only on the $L^2_\unif$-norm of $W$, such that for any  $z\in \CC \setminus \sigma(-\Delta+W)$, we have
\begin{align*}
 \norm{(-\Delta +1)(-\Delta+W-z)^{-1} }_{\B}\leq C\frac{1+\av{z}}{{\rm d}(z,\sigma(-\Delta+W) )}.
\end{align*}
In particular, if $\Lambda$ is a compact set of $\CC \setminus \sigma(-\Delta+W)$, then $(-\Delta +1)(-\Delta+W-z)^{-1}$ is uniformly bounded on $\Lambda$.
\end{lemma}
%
The proof of Lemma~\ref{lemma:B(z)borné} is elementary, it can be read in~\cite{these}.
In view of~\eqref{eq:Y_mfinL2},~\eqref{eq:def_rho} and~\eqref{eq:borneQ1fW}, it follows that 
\begin{align*}
\av{ \int_\RRd (\L f)W} &\leq C\norm{f}_{H^{-1}} \norm{W}_{L^2}.
\end{align*}
We deduce that 
$$
\norm{\L f}_{L^2}\leq C \norm{f}_{H^{-1}}.
$$
We now prove that $(1+\L)^{-1}$ is bounded on $L^2(\RRd)$. Let $g\in L^2(\RRd)$ and $f\in H^{-1}(\RRd)$ such that $(1+\L)f=g$. Then, $f=g-\L f\in L^2(\RRd)$. 
As $1/(1+ \L)$ is bounded from $H^{-1}(\RRd)$ into itself, we have
$$
\norm{f}_{H^{-1}}\leq C\norm{g}_{H^{-1}}\leq C \norm{g}_{L^2}.
$$
Therefore, as $\L$ is continuous from $H^{-1}(\RRd)$ to $L^2(\RRd)$  ,
\begin{align*}
 \norm{f}_{L^2}&=\norm{g-\L f}_{L^2}\leq \norm{g}_{L^2}+\norm{\L f}_{L^2}\leq \norm{g}_{L^2}+ C\norm{f}_{H^{-1}}\leq C \norm{g}_{L^2},
\end{align*}
which concludes the proof of {\it (ii)}.

\paragraph*{Step 3}\textit{Proof of the first part of (iii): $\L$ is well-defined and bounded on $L^2_\unif(\RRd)$. }
 First, we consider a bounded operator $A\in \B(L^2(\RRd))$ and prove that $(z-H_0)^{-1}A(z-H_0)^{-1}$ is locally trace class. For $\chi\in L^\ii_c(\RRd)$ and $z\in \C$, we have by~\eqref{eq:B(z)uniformlybounded} and the Kato-Simon-Seiler inequality~\eqref{eq:KSS} that $\chi(z-H_0)^{-1}A(z-H_0)^{-1}\chi$ is trace class and that there exists $C\geq 0$ independent of $z\in\C$ such that
\begin{align*}
\av{\tr\bra{\chi\frac{1}{z-H_0}A\frac{1}{z-H_0}\chi}}&\leq \norm{\chi\frac{1}{z-H_0}A\frac{1}{z-H_0}\chi}_{ \S_1}\\
&\leq \norm{\chi\frac{1}{z-H_0}}_{\S_2}\norm{A}_{\B}\norm{\frac{1}{z-H_0}\chi}_{\S_2}\leq C \norm{A}_{\B}\norm{\chi}_{L^2}^2.
\end{align*}
It follows that the operator $(z-H_0)^{-1}A(z-H_0)^{-1}$ is locally trace class and that its density $\rho_{z}$ is in $L^1_{\rm loc}(\RRd)$. We now show that $\rho_{z}$ is in fact in $L^2_\unif(\RRd)$. 
Let $k\in \ZZd$ and $u$ be a non-negative function in $L^\ii(\Gamma +k)$. It holds, taking $\chi=\sqrt{u}$, that
\begin{align}\label{eq:rho_zfcontroleeparYf_B}
 \av{\int_\RRd\rho_z u}&=\av{\int_\RRd \rho_z \chi^2}= \av{\tr\bra{\chi\frac{1}{z-H_0}A\frac{1}{z-H_0}\chi}}\leq C \norm{A}_{\B}\norm{u}_{L^1}.
\end{align}
By linearity, we deduce that $\rho_{z}\in L^\ii(\RRd)$ and 
\begin{align*}
 \norm{\rho_z}_{L^2_\unif} \leq \norm{\rho_z}_{L^\ii}\leq C\norm{A}_{\B}.
\end{align*}
As all these estimates are uniform on the compact set $\C$, the operator \\
$\bra{2i\pi}^{-1}\oint_\C \bra{z-H_0}^{-1}A\bra{z-H_0}^{-1}dz$ is locally trace class and its density $\rho$ is in  $L^2_\unif(\RRd)$ and satisfies
\begin{align}\label{eq:rho_Q1fcontroleeparYf_B}
  \norm{\rho}_{L^2_\unif}&\leq C \norm{A}_{\B}.
\end{align}
We now consider the case when $A=Y_m*f$ is a potential generated by a charge density $f\in L^2_\unif(\RRd)$. 
The following Lemma gives the functional space $Y_m*f$ belongs to when $f\in L^2_\unif(\RRd)$.
\begin{lemma}\label{lemma:Young}
Let $f\in L^q_\unif(\RRd)$ and  $Y\in L^p_{\rm loc}(\RRd)$ such that 
\begin{align}\label{eq:hyp_young}
\;\sum_{k\in\ZZd}\norm{Y}_{L^p(\Gamma +k)}<\ii,
\end{align} 
for some $1\leq p, q\leq \ii$. Then, the function $Y*f $ is in $ L^r_\unif(\RRd)$ with $1+1/r=1/p+1/q$ and there exists $C\geq 0$ independent of $f$ such that 
\begin{equation*}
\norm{Y*f}_{L^r_\unif}\leq C \norm{f}_{L^q_\unif}.
\end{equation*} 
\end{lemma}
The proof of Lemma~\ref{lemma:Young} is exactly the same than the one of~\cite[Lemma 3.1]{CaLaLe-12}, we omit it here. As $Y_m$ satisfies~\eqref{eq:hyp_young} for $p=2$, we have 
\begin{equation}\label{eq:Young}
Y_m*f\in L^\ii(\RRd)\quad \text{and}\quad  \norm{Y_m*f}_{L^\ii}\leq C\norm{f}_{L^2_\unif}. 
\end{equation}
Therefore, by~\eqref{eq:rho_Q1fcontroleeparYf_B} 
\begin{equation*}
\norm{\rho_{Q_{1,f}}}_{L^2_\unif}\leq C \norm{Y_m*f}_{L^\ii}\leq C\norm{f}_{L^2_\unif},
\end{equation*}
which proves that  $\L$ is well-defined and bounded from $L^2_\unif(\RRd)$ into itself. This concludes Step 3.

\medskip

In the rest of the proof, we use a localization technique. We will thus need Lemmas~\ref{lemma:estimee_com_L2unif} and~\ref{lemma:type_gronwall} below. Lemma~\ref{lemma:estimee_com_L2unif} gives an estimate on the commutator between the dielectric response operator $\L$ and a localizing function in both $L^2(\RRd)$ and $L^2_\unif(\RRd)$. Lemma~\ref{lemma:type_gronwall} gives a decay rate of a real sequence satisfying a recursion relation that will be satisfied by the localized sequence. The proofs of Lemmas~\ref{lemma:estimee_com_L2unif} and~\ref{lemma:type_gronwall} are postponed until the end of the proof of the proposition.

\begin{lemma}\label{lemma:estimee_com_L2unif}
Let $\chi$ be a smooth function in $C^\ii_c(\RRd)$ such that  $0 \leq \chi\leq 1$, $\chi\equiv 1$ on $B(0,1)$ and $\chi\equiv 0$ outside $B(0,2)$. For any set  $I\subset \ZZd$ and $R\geq 1$ we denote by $B_{I,R}=\cup_{k\in I}\bra{B(0,R)+ k}$ and by $\chi_{I,R}(x)=\chi\bra{{{\rm d}(x, I)}/{R}}$. The family of functions $(\chi_{I,R})_{R\geq 1}$ satisfy $0 \leq \chi_{I,R}\leq 1$, $\chi_{I,R}\equiv 1$ on $B_{I,R}$, $\chi_{I,R}\equiv 0$ outside $B_{I,2R}$ and  
\begin{equation}\label{eq:condition_chi_R}
R\av{\nabla\chi_{I,R}(x)}+R^2\av{\Delta\chi_{I,R}(x)}\leq C\quad\text{a.e.},
\end{equation}
where $C$ is independent of the set $I$. We denote by $\eta_{I,R}=1-\chi_{I,R}$. 
Then, there exists $C\geq 0$ and $C'>0$ such that for any $I\subset \ZZd$ and any $f\in L^2(\RRd)$\footnote{In the whole paper, we use the convention $f*gh=hf*g=h(f*g)$.}
\begin{align}\label{eq:estimee_com_L2}
&\norm{\eta_{I,R}Y_m* f-Y_m*(\eta_{I,R} f)}_{H^2}+\norm{\com{\eta_{I,R},\L} f}_{L^2}\nonumber\\
&\qquad\qquad\qquad\qquad\leq \frac{C}{R}\bra{e^{-C'R}\norm{\1_{(\RRd\setminus B_{I,3R})\cap B_{I,R/2}}f}_{H^{-1}}+\norm{\1_{B_{I,3R}\setminus B_{I,R/2}}f}_{H^{-1}}}\nonumber\\
&\qquad\qquad\qquad\qquad\leq \frac{C}{R}\bra{e^{-C'R}\norm{f}_{L^2}+\norm{\1_{B_{I,3R}\setminus B_{I,R/2}}f}_{L^2}},
\end{align}
and for any $f\in L^2_\unif(\RRd)$
\begin{align}\label{eq:estimee_com_L2unif}
&\norm{\eta_{I,R}Y_m* f-Y_m*(\eta_{I,R} f) }_{H^2_\unif}+\norm{\com{\eta_{I,R},\L} f}_{L^2_\unif}\nonumber\\
&\qquad\qquad\qquad\qquad\leq \frac{C}{R}\bra{e^{-C'R}\norm{f}_{L^2_\unif}+\norm{ \1_{B_{I,3R}\setminus B_{I,R/2}} f}_{L^2_\unif}}.
\end{align}
\end{lemma}

\begin{lemma}\label{lemma:type_gronwall}
 Let $(x_R)_{R\geq 0}$  be a non-increasing family of real numbers such that for any $R>0$, 
\begin{equation}\label{eq:assy_type_gron}
x_R\leq \frac{C}{R}e^{-C'R}x_0+\frac{C}{R}x_{R/a}
\end{equation}
for given $C\geq 0$ and $C',a>0$. Then, there exists $C\geq 0$ and $C'>0 $ such that for any $R\geq 2$
\begin{equation}\label{eq:decroissance_logarithmique}
x_R\leq Ce^{-C'(\log R)^2}x_0.
\end{equation}
\end{lemma}

We now proceed with the proof of Theorem~\ref{th:props_de_L}. We first prove {\it (iv)}, then we prove that $1+\L$ is invertible on $L^2_\unif(\RRd)$. 

\paragraph*{Step 4}\textit{Proof of {\it (iv)}.} 
We explain how to use Lemmas~\ref{lemma:estimee_com_L2unif} and~\ref{lemma:type_gronwall} to prove~\eqref{eq:1/1+L_localise_unif}. Let $k\in \ZZd$ and for $R\geq 1$, let $\eta_R=\eta_{\set{k},R}$ and $B_R=B_{\set{k},R}$ as defined in Lemma~\eqref{lemma:estimee_com_L2unif}.
Let $g\in L^2(\RRd)$ and denote by $f=\bra{1+\L}^{-1}\1_{\Gamma+k}g$. 
For $R\geq 1$, we have
$$
\eta_R\bra{f + \L f} =\eta_R \1_{\Gamma+k} g =0.
$$
Therefore
$$
\bra{1+\L}\eta_R f=\eta_Rf+ \L\eta_Rf = \L\eta_Rf -\eta_R \L f=\com{\L,\eta_R}f.
$$
Since $1/(1+\L)$ is bounded on $L^2(\RRd)$, it follows that
\begin{align}\label{eq:decroissance_etaRf}
\norm{\eta_Rf}_{L^2}= \norm{\frac{1}{1+\L}\com{\L,\eta_R} f}_{L^2}&\leq C\norm{\com{\L,\eta_R}  f}_{L^2}\nonumber\\
&\leq \frac{C}{R}e^{-C'R}\norm{f}_{L^2}+\frac{C}{R}\norm{\1_{B_{3R}\setminus B_{R/2}}f}_{L^2},
\end{align}
where we have used Lemma~\ref{lemma:estimee_com_L2unif} in the last step. Denoting by $x_R=\norm{\1_{\RRd\setminus B_{2R}}f}_{L^2}$, the estimate~\eqref{eq:decroissance_etaRf} leads to
$$
x_R\leq \frac{C}{R}e^{-C'R}x_0+ \frac{C}{R}x_{R/4}.
$$
Therefore, Lemma~\ref{lemma:type_gronwall} gives that there exists $C\geq 0$ and $C'>0$ such that for any $R\geq 2$
\begin{align*}
\norm{\eta_Rf }_{L^2}\leq x_{R/2}\leq Ce^{-C'\bra{\log R}^2}x_0 = C e^{-C'\bra{\log R}^2}\norm{f}_{L^2}&\leq C e^{-C'\bra{\log R}^2}\norm{g}_{L^2(\Gamma+k)},
\end{align*}
where the last inequality follows from the fact that $(1+\L)^{-1}$ is bounded on $L^2(\RRd)$. Finally, as $\1_{\Gamma+j}\leq \eta_{\av{k-j}/1-1/2}$, then
\begin{equation*}
\norm{\1_{\Gamma+j} \frac{1}{1+\L}\1_{\Gamma+k}g }_{L^2}\leq C e^{-C'\bra{\log \av{k-j}}^2}\norm{g}_{L^2(\Gamma+k)}\leq C e^{-C'\bra{\log \av{k-j}}^2}\norm{g}_{L^2}.
\end{equation*}

\paragraph*{Step 5}\textit{Proof that $1+ \L$ is surjective on $L^2_\unif(\RRd)$.}  Let $g\in L^2_\unif(\RRd)$ and consider $g_L=g\1_{{\Gamma} _L}$ for $L\in 2\NN+1$. As $1+\L$ is invertible on $L^2(\RRd)$, there exists $f_L\in L^2(\RRd)$ such that 
\begin{equation}\label{eq:equation_L}
(1+\L)f_L=g_L
\end{equation} 
and 
\begin{align*}
 \norm{f_L}_{L^2_\unif}&=\sup_{j\in \ZZd}\norm{\1_{\Gamma +j}\frac{1}{1+\L}\sum_{k\in \ZZd\cap {\Gamma} _L}\1_{\Gamma +k}g}_{L^2}\leq\sup_{j\in\ZZd } \sum_{k\in\ZZd\cap \Gamma_L}\norm{\1_{\Gamma +j}\frac{1}{1+\L}\1_{\Gamma +k}g}_{L^2}.
\end{align*}
Using~\eqref{eq:1/1+L_localise_unif},  we obtain
\begin{align*}
\norm{f_L}_{L^2_\unif}&\leq\sup_{j\in \ZZd} C\sum_{k\in\ZZd\setminus \set{j}}e^{-C'\bra{\log\av{j-k}}^2} \norm{g}_{L^2(\Gamma +k)}+C\norm{g}_{L^2_\unif}\leq C\norm{g}_{L^2_\unif}
\end{align*}
for a constant $C$ independent of $L$. The space $L^2_\unif(\RRd)$ is known to be the dual of  $\ell ^1(L^2)=\set{f\in L^2_\text{loc}(\RRd),\; \sum_{k\in\ZZd}\norm{f}_{L^2(\Gamma +k)}<\ii}$, which is a separable Banach space. Therefore, since the sequence $(f_L)_{L\geq 1}$ is bounded in  $L^2_\unif(\RRd)$, there exists a subsequence of  $(f_L)_{L\geq 1}$ (denoted the same for simplicity)
and $f\in L^2_\unif(\RRd)$  such that $f_{L} \rightharpoonup _\ast f$ in $L^2_\unif(\RRd)$ and 
\begin{align}\label{eq:1/1+L_borne_L2unif}
\norm{f}_{L^2_\unif}\leq \liminf_{k\rightarrow \ii} \norm{f_{L}}_{L^2_\unif}\leq  C \norm{g}_{L^2_\unif}.
\end{align} 
We now want to pass to the limit in the sense of distributions in~\eqref{eq:equation_L}. Since $C^\ii_c(\RRd)$ is dense in $\ell^1(L^2) $, the sequence  $(f_L)$ converges to $f$ in $\D'(\RRd)$. Next, we need to show that for any $\phi\in \D(\RRd)$,
\begin{equation}\label{eq:weak_continuity_L}
\int_\RRd\bra{\L\bra{f_L-f}}\phi\cvL 0.
\end{equation}
We denote by $\rho_{z,L}$ the density associated with the operator $\bra{z-H_0}^{-1} Y_m*(f-f_L)\bra{z-H_0}^{-1}$. Then
\begin{align*}
 \int_\RRd\bra{\L\bra{f_L-f}}\phi& =\frac{1}{2i\pi}\oint_\C \int_\RRd\rho_{z,L}\phi \,dz
\end{align*}
and, as $\phi$ has compact support, we have by~\eqref{eq:rho_zfcontroleeparYf_B} and~\eqref{eq:Young}
\begin{align*}
\av{\int_ \RRd\rho_{z,L}\phi}
&\leq C \norm{f-f_L}_{L^2_\unif}\norm{\phi}_{L^1}
\leq C \norm{g}_{L^2_\unif}\norm{\phi}_{L^1},
\end{align*}
where the constant $C\geq 0$ is independent of $L$ and $z\in \C$. By the dominated convergence theorem, it is therefore sufficient, for proving~\eqref{eq:weak_continuity_L}, to show that for any $z\in \C$
\begin{equation}\label{eq:rho_z}
\int_{\RRd}\rho_{z,L}\phi\cvL 0.
\end{equation}
 For $R\geq 1$, we define  $\rho_{z,L,{\rm out},R}$ and $\rho_{z,L,{\rm in} ,R}$ to be the densities associated with the operators
$$
\frac{1}{z-H_0}\1_{\RRd\setminus B(0,R)}Y_m*(f-f_L)\frac{1}{z-H_0}\quad \text{and}\quad \frac{1}{z-H_0}\1_{ B(0,R)}Y_m*(f-f_L)\frac{1}{z-H_0}
$$
respectively. Therefore $\rho_{z,L}=\rho_{z,L,\rm out,R}+ \rho_{z,L,{\rm in},R}$. Let $\epsilon >0$. In the following, we will choose $R$ large enough such that $\int \rho_{z,L,{\rm out},R}\phi$ is small for any $L$. Then, using the weak-$\ast$ convergence of  $f_L$ to $f$ we show that $\int \rho_{z,L,{\rm in},R}\phi$ is small for $L$ large enough. Reasoning similarly than in the proof of~\eqref{eq:rho_zfcontroleeparYf_B}, we find
\begin{align}\label{eq:3yit}
\av{\int_\RRd\rho_{z,L,{\rm out},R}\phi}&\leq C \norm{f-f_L}_{L^2_\unif}\\
&\times\bra{\norm{ \1_{\RRd\setminus B(0,R)}\frac{1}{z-H_0}\sqrt{\phi_+}}_{\S_2}^2+\norm{\1_{\RRd\setminus B(0,R)}\frac{1}{z-H_0}\sqrt{\phi_-}}_{\S_2}^2}\nonumber
\end{align}
Now, we need the following lemma.
\begin{lemma}\label{lemma:CT}
Let $W\in L^2_\unif(\RRd)$ and $H=-\Delta+W$. There exists $C\geq 0$ and $C'>0$, depending only on $\norm{W}_{L^2_\unif}$, such that for any $\chi\in L^2(\RRd)$ and $\eta\in L^\ii(\RRd)$ satisfying $R={\rm d}\bra{\text{supp}(\chi),\text{supp}(\eta)}\geq 1$, and any $z\in \CC\setminus \sigma(H)$, we have 
\begin{equation*}
 \norm{\chi\bra{z-H}^{-1}\eta}_{\S_2}\leq Cc_1(z)e^{-C'c_2(z)R}\norm{\eta}_{L^\ii}\norm{\chi}_{L^2},
\end{equation*}
where $c_1(z)={\rm d}(z, \sigma( H))^{-1}$, $c_2(z)={\rm d}(z,\sigma(H))/(\av{z}+1)$.
In particular, if $\Lambda$ is a compact set of $\CC \setminus \sigma(H)$, then 
\begin{equation*}
 \norm{\chi\bra{z-H}^{-1}\eta}_{\S_2}\leq Ce^{-C'R}\norm{\eta}_{L^\ii}\norm{\chi}_{L^2},
\end{equation*}
where $C$ and $C'$ do not depend on $z$ but depend, in general, on $\Lambda$. 
\end{lemma}

\begin{proof}[Proof of Lemma~\ref{lemma:CT}]
 We have
\begin{align*}
 \norm{\chi\bra{z-H}^{-1}\eta}_{\S_2}^2=\int_{\RRd\times \RRd}\av{\chi(x)G_z(x,y)\eta(y)}^2\,dx\,dy,
\end{align*}
where $G_z(x,y)$ in the kernel of $\bra{z-H}^{-1}$. By~\cite[Theorem B.7.2]{simon_SSG} and~\cite[Corollary 1]{germinet-Klein-Decay} we have for $\av{x-y}\geq 1$
\begin{align*}
 \av{G_z(x,y)}\leq Cc_1(z)e^{-C'c_2(z)\av{x-y}},
\end{align*}
where $C\geq 0$ and $C'>0$ depend only on $\norm{W}_{L^2_\unif}$. Therefore 
\begin{align*}
 \norm{\chi\bra{z-H}^{-1}\eta}_{\S_2}^2&\leq Cc_1(z)^2\norm{\eta}_{L^\ii}^2 \norm{\chi}_{L^2}^2 \sup_{x\in \text{supp}(\chi)}\int_{\RRd}\1_{\text{supp}(\eta)}(y)e^{-2C'c_2(z)\av{x-y}}\,dy\\
&\leq Cc_1(z)^2\norm{\eta}_{L^\ii}^2 \norm{\chi}_{L^2}^2 e^{-C'c_2(z)R}.
\end{align*} 
\end{proof}
We now go back to~\eqref{eq:3yit}. Using that $\norm{f-f_L}_{L^2_\unif}\leq C\norm{g}_{L^2_\unif}$ and Lemma~\ref{lemma:CT}, we have for $R$ large enough
\begin{align}\label{eq:weak_continuity}
&\av{ \int_\RRd\rho_{z,L,{\rm out},R}\phi } \leq C\norm{g}_{L^2_\unif}\norm{\phi}_{L^1}^2e^{-C'R}.
\end{align}
We can thus choose $R$  such that~\eqref{eq:weak_continuity} is smaller than $\epsilon/2$. Besides, we have
\begin{align*}
\int_\RRd \rho_{z,L,{\rm in},R}\phi&= \tr\bra{\1_{B(0,R)}Y_m*(f-f_L)\frac{1}{z-H_0}\phi \frac{1}{z-H_0}}\\
&=\int_\RRd \1_{B(0,R)}Y_m*(f_L-f) \rho,
\end{align*}
where $\rho$ is the density associated with the trace class operator $\bra{z-H_0}^{-1}\phi \bra{z-H_0}^{-1}$. For $R'>0$, we have
\begin{align}\label{eq:term_chiR}
& \av{ \int_\RRd \1_{B(0,R)}  Y_m*(f_L-f)\rho}= \av{\int_{B(0,R)} \int_{\RRd}Y_m(x-y)\bra{f-f_L}(y)\,dy \rho(x)\,dx}\nonumber\\
&\qquad\qquad\qquad\leq \av{ \int_{B(0,R)} \int_{B(0,R')}Y_m(x-y)\bra{f-f_L}(y)\,dy \rho(x)\,dx}\nonumber\\
&\qquad\qquad\qquad+\norm{\int_{\RRd\setminus B(0,R')}Y_m(\cdot-y)\bra{f-f_L}(y)\,dy}_{L^\ii(B(0,R))} \norm{\rho}_{L^1}.
\end{align}
As $Y_m$ is exponentially decaying, we can choose $R'$ such that the second term of the RHS of~\eqref{eq:term_chiR} is smaller that $\epsilon/4$. 
As to the first term, by the weak-$\ast$ convergence of $f_L$ to $f$ in $L^2_\unif(\RRd)$, we have that 
$$h_L(x)=\int_{B(0,R')} Y_m(x-y)\bra{f-f_L}(y)\,dy \quad{\cvL}\quad 0,$$
for any $x\in B(0,R)$. Besides, we have for a.e. $x\in B(0,R)$
$$
\av{h_L(x)}\leq \norm{h_L}_{L^\ii}\leq C\norm{f-f_L}_{L^2_\unif}\leq C\norm{g}_{L^2_\unif}
$$
(see~\eqref{eq:Young}). By the dominated convergence theorem, it follows that one can choose $L$ large enough such that the first term of  the RHS of~\eqref{eq:term_chiR} is smaller that $\epsilon/4$. This concludes the proof of~\eqref{eq:rho_z}, thus the proof of \eqref{eq:weak_continuity_L}. We are now able to pass to the limit in~\eqref{eq:equation_L}, which concludes the proof of the surjectivity of $1+\L$ on $L^2_\unif(\RRd)$. In view of ~\eqref{eq:1/1+L_borne_L2unif}, we have shown that 
there exists $C\geq 0$ such that for any  $g\in L^2_\unif(\RRd)$, there exists $f\in L^2_\unif(\RRd)$ such that 
\begin{equation}\label{eq:1/1+L_borne}
(1+\L)f=g\quad \text{and}\quad \norm{f}_{L^2_\unif}\leq C \norm{g}_{L^2_\unif}.
\end{equation}

\paragraph*{Step 6}\textit{Proof that $1+\L$ is injective on $L^2_\unif(\RRd)$.}
Let $f\in L^2_\unif(\RRd)$ be such that $(1+\L)f=0$. For $R\geq 1$, let $\chi_R=\chi_{\set{0},R}$ as in Lemma~\ref{lemma:estimee_com_L2unif}. Then,
\begin{align*}
\chi_R f+\chi_R\L(f)=0,
\end{align*}
and thus
\begin{align*}
 \bra{1+\L}(\chi_Rf)=\L\chi_Rf-\chi_R\L(f)=\com{\L,\chi_R}f.
\end{align*}
As $g:=\com{\L,\chi_R}f\in L^2(\RRd)$, then the solution $\phi=\chi_Rf$ of $(1+\L)\phi=g$ is unique and satisfies $\norm{\phi}_{L^2_\unif}\leq C\norm{g}_{L^2_\unif}$ by~\eqref{eq:1/1+L_borne}. Therefore
$$
\norm{\chi_Rf}_{L^2_\unif}\leq C\norm{\com{\L,\chi_R}f}_{L^2_\unif}.
$$
Using Lemma~\ref{lemma:estimee_com_L2unif}, we have
\begin{align}\label{eq:chiRf_majoree_1/R}
 \norm{\chi_Rf}_{L^2_\unif}&\leq C\norm{\com{\L,\chi_R}f}_{L^2_\unif}=C\norm{\com{\L,\eta_R}f}_{L^2_\unif}\leq\frac{C}{R} \norm{f}_{L^2_\unif}.
\end{align}
As $\norm{\chi_Rf}_{L^2_\unif}$ is a non-decreasing function of $R$ converging to $\norm{f}_{L^2_\unif}$ when $R\rightarrow+\ii$ and the RHS of~\eqref{eq:chiRf_majoree_1/R} goes to $0$ when $R\rightarrow+\ii$, then $\norm{f}_{L^2_\unif}=0$ and $f=0$; which proves that $1+\L$ is injective. The boundedness of $1/(1+\L)$ then follows from~\eqref{eq:1/1+L_borne}.
This  concludes the proof of Theorem~\ref{th:props_de_L}.
\end{proof}

In order to complete the proof of Theorem~\ref{th:props_de_L}, we need to prove Lemmas~\ref{lemma:estimee_com_L2unif} and~\ref{lemma:type_gronwall}.

\begin{proof}[Proof of Lemma~\ref{lemma:estimee_com_L2unif}]
For simplicity, we use the shorthand notation $\chi_R=\chi_{I,R}$, $\eta_R=\eta_{I,R}$ and $B_R=B_{I,R}$.

\paragraph*{Step 1}\textit{Proof of~\eqref{eq:estimee_com_L2}.} 
We have
\begin{align*}
 \eta_R f*Y_m -Y_m*\bra{\eta_Rf}&=\eta_R\bra{-\Delta+m^2}^{-1}f-\bra{-\Delta+m^2}^{-1}\eta_Rf\\
&= \com{\eta_R,\bra{-\Delta+m^2}^{-1}}f.
\end{align*}
We now use that $\com{B, (z-A)^{-1}}=(z-A)^{-1}\com{B,A} (z-A)^{-1}$ and the fact that $\com{\eta_R, \Delta}=-(\Delta \eta_R+2\nabla\eta_R\cdot \nabla)$. We thus obtain
\begin{align}\label{eq:com_V}
\eta_R f*Y_m -Y_m&*\bra{\eta_Rf}= \bra{-\Delta+m^2}^{-1}\com{\eta_R, \Delta}\bra{-\Delta+m^2}^{-1}f\nonumber\\
&= -\bra{-\Delta+m^2}^{-1}\bra{(\Delta \eta_R)+2\bra{\nabla \eta_R}\cdot \nabla }\bra{-\Delta+m^2}^{-1}f.
\end{align}
As $\nabla\eta_R=-\nabla \chi_R$ and $\Delta \eta_R=-\Delta \chi_R$ are supported in $B_{2R}\setminus B_R$, then, by~\eqref{eq:condition_chi_R},
\begin{align}\label{eq:etaV-Veta}
 \norm{ \eta_R f*Y_m -Y_m*\bra{\eta_Rf}}_{H_2}&\leq \frac{C}{R^2}\norm{\1_{B_{2R}\setminus B_{R}}\bra{-\Delta+m^2}^{-1}f}_{L^2}\nonumber\\
&+ \frac{C}{R}\norm{\1_{B_{2R}\setminus B_{R}}\nabla\bra{-\Delta+m^2}^{-1}f}_{\bra{L^2(\RRd)}^d}.
\end{align}
To bound the first term of the RHS of~\eqref{eq:etaV-Veta}, we write
\begin{align}\label{eq:diviser_V}
\bra{-\Delta+m^2}^{-1}f (x)
&=\int_\RRd Y_m(x-y)f(y)\1_{(\RRd\setminus B_{3R})\cup B_{R/2}  }(y)\,dy\nonumber\\
&\qquad+\int_\RRd Y_m(x-y)f(y)\1_{B_{3R}\setminus B_{R/2}}(y)\,dy
\end{align}
Thanks to the exponential decay of $Y_m$ and the fact that for any $x\in B_{2R}\setminus B_{R}$ and $y\in (\RRd\setminus B_{3R})\cup B_{R/2}$, $\av{x-y}\geq R/2$, we get
\begin{align*}
&\norm{\1_{B_{2R}\setminus B_{R}} \bra{-\Delta+m^2}^{-1}f\1_{ (\RRd\setminus B_{3R} )\cup B_{R/2}}}_{L^2}\\
&\qquad\qquad\leq Ce^{-\frac{mR}{4}}\norm{Y_{\frac{m}{2}}*(f\1_{ (\RRd\setminus B_{3R} )\cup B_{R/2}})}_{L^2} \leq Ce^{-\frac{mR}{4}}\norm{f\1_{ (\RRd\setminus B_{3R} )\cup B_{R/2}}}_{H^{-2}}.
\end{align*}
Controlling in the same way the second term of the RHS of~\eqref{eq:diviser_V}, we deduce
\begin{align*}
\norm{\1_{B_{2R}\setminus B_{R}} \bra{-\Delta+m^2}^{-1}f}_{L^2}\leq Ce^{-\frac{mR}{4}}\norm{\1_{ (\RRd\setminus B_{3R} )\cup B_{R/2}}f}_{H^{-2}}+C\norm{\1_{B_{3R}\setminus B_{R/2}}f}_{H^{-2}}.
\end{align*}
We proceed similarly for the second term of the RHS of~\eqref{eq:etaV-Veta} using that $W_m=\nabla Y_m$, the inverse Fourier transform of  $i\av{S^{d-1}}\frac{p}{\av{p}^2
+m^2}$, is exponentially decaying and satisfies $\norm{W_m*g}_{L^2}\leq \norm{g}_{H^{-1}} $ for any $g\in H^{-1}$. We get
\begin{equation}\label{eq:r1_f}
\norm{ \eta_R f*Y_m -Y_m*\bra{\eta_Rf}}_{H_2}\leq \frac{C}{R}e^{-\frac{mR}{4}}\norm{\1_{ (\RRd\setminus B_{3R} )\cup B_{R/2}}f}_{H^{-1}}+\frac{C}{R}\norm{\1_{B_{3R}\setminus B_{R/2}}f}_{H^{-1}}.
\end{equation}

We turn now to estimating $\norm{\com{\eta_R,\L} f}_{L^2}$. We know that $\com{\eta_R,\L} f$ is the density associated with the operator
\begin{align}\label{eq:decomposition_com}
 -\eta_RQ_{1,f}+Q_{1, \eta_Rf} &=\frac{1}{2i\pi}\int_\C \bra{ \frac{1}{z-H_0}Y_m*(\eta_Rf)\frac{1}{z-H_0}-\eta_R\frac{1}{z-H_0}Y_m*f\frac{1}{z-H_0}}dz\nonumber\\
&=\frac{1}{2i\pi} \int_\C \frac{1}{z-H_0}\bra{Y*(\eta_Rf)-\eta_RY_m*f}\frac{1}{z-H_0}\,dz\nonumber\\
&\quad-\frac{1}{2i\pi} \int_\C \com{\eta_R,\frac{1}{z-H_0}}Y_m*f\frac{1}{z-H_0}\,dz.
\end{align}
We denote by $r_1$ and $r_2$ the densities associated with the first and second terms of the RHS of~\eqref{eq:decomposition_com} respectively. For any $W\in L^2(\RRd)$, we have
\begin{align}\label{eq:r_1W}
\av{ \int_{\RRd}r_1W}&=\av{\frac{1}{2i\pi} \int_\C \tr\bra{\frac{1}{z-H_0}\bra{Y_m*(\eta_Rf)-\eta_RY_m*f}\frac{1}{z-H_0}W}\,dz}\nonumber\\
&\leq C \norm{ Y_m*(\eta_Rf)-\eta_RY_m*f}_{L^2}\norm{W}_{L^2},
\end{align}
where we have used~\eqref{eq:KSS} and~\eqref{eq:B(z)uniformlybounded}. Therefore, in view of~\eqref{eq:r1_f},
\begin{align}\label{eq:r1_majore_V}
\norm{r_1}_{L^2}&\leq C\norm{Y_m*(\eta_R f)-\eta_RY_m*f}_{L^2}\nonumber\\
&\leq \frac{C}{R}e^{-\frac{mR}{4}}\norm{1_{(\RRd\setminus B_{3R})\cup B_{R/2}} f}_{H^{-1}}+\frac{C}{R}\norm{1_{B_{3R}\setminus B_{R/2}}f}_{H^{-1}}.
\end{align}
It remains to estimate $r_2$. For any $A\in \S_2(L^2(\RRd))$ and $W\in L^2(\RRd)$, the density $\rho$ associated with the operator $\bra{-\Delta+1}^{-{1}/{2}}A\bra{-\Delta+1}^{-{1}/{2}}$ satisfies 
\begin{align*}
\av{ \int_\RRd \rho W} &\leq \norm{\sqrt{\av{W}}}_{L^4}\norm{\bra{\av{p}^2+1}^{-\frac{1}{2}}}_{L^4}\norm{A}_{\S_2}\norm{\bra{\av{p}^2+1}^{-\frac{1}{2}}}_{L^4}\norm{\sqrt{\av{W}}}_{L^4}\nonumber\\
&\leq C \norm{W}_{L^2}\norm{A}_{\S_2}.
\end{align*}
 Therefore 
\begin{align}\label{eq:generic_A}
\norm{\rho}_{L^2}\leq C \norm{A}_{\S_2}.
\end{align}
Applying~\eqref{eq:generic_A} for $A=\bra{-\Delta+1}^{1/2}\com{\eta_R,(z-H_0)^{-1}}Y_m*f(z-H_0)^{-1}\bra{-\Delta+1}^{1/2}$, we obtain
\begin{align}\label{eq:r_2_debut}
 \norm{r_2}_{L^2}&\leq C\oint_\C\norm{\bra{-\Delta+1}^{-\frac{1}{2}}\com{\eta_R, \Delta}\frac{1}{z-H_0}Y_m*f\bra{-\Delta+1}^{-\frac{1}{2}}}_{\S_2}dz,
\end{align}
where we have used that $C_1(1-\Delta)\leq \av{z-H_0}\leq C_2 (1-\Delta)$, whose proof is similar to the the one of Lemma~\ref{lemma:B(z)borné}. 
As the commutator $\com{\eta_R, \Delta}$ has its support in $B_{2R}\setminus B_R$, we consider separately $f\1_{B_{3R}\setminus B_{R/2} }$ and $f\1_{(\RRd\setminus B_{3R})\cup B_{R/2}}$.
Using the same techniques as above, we obtain
\begin{align}\label{eq:r_2_part1}
& \norm{ \bra{-\Delta+1}^{-\frac{1}{2}}\com{\eta_R, \Delta}\frac{1}{z-H_0}Y_m*\bra{\1_{B_{3R}\setminus B_{R/2} }f}\bra{-\Delta+1}^{-\frac{1}{2}}}_{\S_2}\nonumber\\
&\qquad\qquad\qquad\qquad\leq C\norm{\bra{-\Delta+1}^{-\frac{1}{2}}\com{\eta_R,\Delta}}_{\B}\norm{\frac{1}{z-H_0}Y_m*\bra{1_{B_{3R}\setminus B_{R/2} }f}}_{\S_2}\nonumber\\
&\qquad\qquad\qquad\qquad\leq C\bra{\norm{\nabla\eta_R }_{L^\ii}+\norm{\Delta\eta_R }_{L^\ii}}\norm{Y_m*\bra{1_{B_{3R}\setminus B_{R/2} }f}}_{L^2}\nonumber\\
&\qquad\qquad\qquad\qquad\leq \frac{C}{R}\norm{\1_{B_{3R}\setminus B_{R/2} }f}_{H^{-2}}.
\end{align}
Far from the support of $\com{\eta_R, \Delta}$, we have
\begin{align}\label{eq:sum}
& \norm{ \bra{-\Delta+1}^{-\frac{1}{2}}\com{\eta_R, \Delta}\frac{1}{z-H_0}Y_m*\bra{\1_{(\RRd\setminus B_{3R})\cup B_{R/2}}f}\bra{-\Delta+1}^{-\frac{1}{2}}}_{\S_2}\nonumber\\
&\qquad\qquad\qquad\leq C\norm{ \bra{-\Delta+1} ^{-\frac12} \bra{\Delta \eta_R-2\nabla \cdot\nabla\eta_R} }_{\B}\sum_{k\in\ZZd}\norm{\1_{B_{2R}\setminus {B_R}} \frac{1}{z-H_0}\1_{\Gamma +k}}_{\S_2}\nonumber\\
&\qquad\qquad\qquad\quad\times\norm{\1_{\Gamma +k}Y_m*\bra{\1_{(\RRd\setminus B_{3R})\cup B_{R/2}}f}\bra{-\Delta+1} ^{-\frac12} }_{\B}.
\end{align}
In dimension $d\leq 3$, $H^1(\RRd)\hookrightarrow L^4(\RRd)$. Therefore
\begin{align}\label{eq:deBàH1}
 &\norm{\1_{\Gamma +k}Y_m*\bra{\1_{(\RRd\setminus B_{3R})\cup B_{R/2}}f}\bra{-\Delta+1} ^{-\frac12} }_{\B}\nonumber\\
&\qquad\qquad\leq C \norm{\1_{\Gamma +k}Y_m*\bra{\1_{(\RRd\setminus B_{3R})\cup B_{R/2}}f}\bra{-\Delta+1} ^{-\frac12} }_{\S_4}\nonumber\\
&\qquad\qquad \leq C \norm{\1_{\Gamma +k}Y_m*\bra{\1_{(\RRd\setminus B_{3R})\cup B_{R/2}}f} }_{L^4} \leq C\norm{Y_m*\bra{\1_{(\RRd\setminus B_{3R})\cup B_{R/2}}f} }_{H^1(\Gamma +k)}.
\end{align}
Using the exponential decay of $Y_m$, we obtain 
\begin{align}\label{eq:yellow0}
\norm{Y_m*\bra{\1_{(\RRd\setminus B_{3R})\cup B_{R/2}}f} }_{H^1(\Gamma +k)}&\leq Ce^{-\frac{m}{2}{\rm d}(k, (\RRd\setminus B_{3R})\cup B_{R/2})}\norm{1_{(\RRd\setminus B_{3R})\cup B_{R/2}}f}_{H^{-1}}.
\end{align}
In particular, for $k\in \ZZd\cap \bra{B_{5R/2}\setminus B_{3R/4}}$ (the pink part in Figure~\ref{fig:lalignereelle} below), the distance between $k$ and  $(\RRd\setminus B_{3R})\cup B_{R/2}$ (the blue part in Figure~\ref{fig:lalignereelle}) is greater than or equal to $R/4$ and
\begin{figure}[!h]
\centering
\psfrag{0}{$0$} \psfrag{R/2}{$\frac{R}{2}$} \psfrag{3R/4}{$\frac{3R}{4}$} \psfrag{R}{$R$} \psfrag{2R}{$2R$} \psfrag{5R/2}{$\frac{5R}{2}$} \psfrag{3R}{$3R$} 
\includegraphics[scale=0.6]{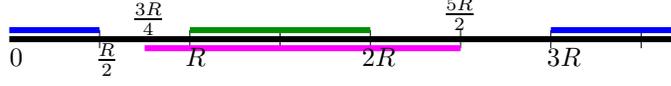}
\caption{Schematic representation of $\RR_+$ used in the proof of Lemma~\ref{lemma:estimee_com_L2unif}.}
\label{fig:lalignereelle}
\end{figure}

\begin{align}\label{eq:yellow}
\norm{Y_m*\bra{\1_{(\RRd\setminus B_{3R})\cup B_{R/2}}f} }_{H^1(\Gamma +k)}&\leq Ce^{-\frac{mR}{16}}e^{-\frac{m}{4}{\rm d}(k, (\RRd\setminus B_{3R})\cup B_{R/2})}\norm{1_{(\RRd\setminus B_{3R})\cup B_{R/2}}f}_{H^{-1}}.
\end{align}
Besides, using Lemma~\ref{lemma:CT} with $\eta=\1_{B_{2R}\setminus B_{R}}$ and $\chi=\1_{\Gamma +k}$, we obtain 
\begin{align*}
\norm{\1_{B_{2R}\setminus {B_R}} \frac{1}{z-H_0}\1_{\Gamma +k}}_{\S_2}&\leq Ce^{-C'{\rm d}(k, B_{2R}\setminus  B_{R})}.
\end{align*}
In particular for $k\in\ZZd\setminus \bra{B_{5R/2}\setminus B_{3R/4}}$, we have ${\rm d}(k,  B_{2R}\setminus  B_{R})\geq \frac{R}{4}$  (see Figure~\ref{fig:lalignereelle}) and  
\begin{align}\label{eq:yellow_bar}
\norm{\1_{B_{2R}\setminus_{B_R}} \frac{1}{z-H_0}\1_{\Gamma +k}}_{\S_2}&\leq Ce^{-\frac{C'}{2}\frac{R}{4}}e^{-\frac{C'}{2}{\rm d}(k, B_{2R}\setminus B_{R})}.
\end{align}
Combining~\eqref{eq:sum},~\eqref{eq:deBàH1},~\eqref{eq:yellow0},~\eqref{eq:yellow} and~\eqref{eq:yellow_bar}, we obtain
\begin{align}\label{eq:r_2_part2}
&\norm{ \bra{-\Delta+1}^{-\frac12} \com{\eta_R,\Delta}\frac{1}{z-H_0}Y_m*\bra{\1_{(\RRd\setminus B_{3R})\cup B_{R/2}}f}\bra{-\Delta+1}^{-\frac12}}_{\S_2} \nonumber\\
&\qquad\qquad \qquad\leq\frac{C}{R}e^{-C'R}\sum_{k\in\ZZd}e^{-C'\av{k}}\norm{1_{(\RRd\setminus B_{3R})\cup B_{R/2}}f}_{H^{-1}}\nonumber\\
&\qquad\qquad \qquad\leq\frac{C}{R}e^{-C'R}\norm{1_{(\RRd\setminus B_{3R})\cup B_{R/2}}f}_{H^{-1}}.
\end{align}
This completes our estimate on $r_2$. Indeed, in view of~\eqref{eq:r_2_debut},~\eqref{eq:r_2_part1} and~\eqref{eq:r_2_part2}, we deduce that
\begin{align*}
 \norm{r_2}_{L^2}\leq \frac{C}{R}e^{-C'R}\norm{1_{(\RRd\setminus B_{3R})\cup B_{R/2}}f}_{H^{-1}}+ \frac{C}{R}\norm{\1_{B_{3R}\setminus B_{R/2} }f}_{H^{-1}},
\end{align*}
which concludes the proof of~\eqref{eq:estimee_com_L2}.

\paragraph*{Step 2}\textit{Proof of~\eqref{eq:estimee_com_L2unif}.}
The proof of~\eqref{eq:estimee_com_L2unif} for functions in $L^2_\unif$ is similar to the one of~\eqref{eq:estimee_com_L2} for $L^2$ functions. We sketch here the main steps of the proof, and only highlighting the differences. Let $f\in L^2_\unif(\RRd)$. Using~\eqref{eq:com_V}, we have
\begin{align*}
 \eta_R Y_m*f-Y_m*(\eta_R f) &= \sum_{k\in \ZZd} \eta_R Y_m*(\1_{\Gamma +k} f )- Y_m*(\eta_R \1_{\Gamma +k}  f)\\
&= \sum_{k\in \ZZd} (-\Delta+m^2)^{-1} \bra{(\Delta\eta_R)+ 2(\nabla\eta_R)\cdot\nabla}(-\Delta+m^2)^{-1}\1_{\Gamma +k}f\\
&=Y_m*\bra{ \Delta\eta_R Y_m*f+2\nabla \eta_R \cdot \nabla Y_m*f}.
\end{align*}
Therefore 
\begin{align}\label{eq:com_V_unif}
\norm{ \eta_R Y_m*f-Y*(\eta_R f)}_{H^2_\unif}& \leq  C\norm{ (-\Delta+m^2) \bra{\eta_R Y_m*f-Y_m*(\eta_R f)}}_{L^2_\unif}\nonumber\\
&\leq   C \norm{ \Delta\eta_R Y_m*f+ 2\nabla \eta_R \cdot\nabla Y_m*f }_{L^2_\unif}\nonumber\\
&\leq \frac{C}{R^2}\norm{\1_{B_{2R}\setminus B_R}Y_m*f }_{L^2_\unif}+ \frac{C}{R}\norm{\1_{B_{2R}\setminus B_R}\nabla Y_m*f }_{L^2_\unif}.
\end{align}
To bound the first term of the RHS of~\eqref{eq:com_V_unif}, we use the exponential decay of $Y_m$, the fact that $Y_m\in \ell^1(L^1)$ and Lemma~\ref{lemma:Young}. We get
\begin{align*}
 \norm{\1_{B_{2R}\setminus B_R}Y_m*f }_{L^2_\unif}
&\leq e^{-\frac{mR}{4}}\norm{Y_\frac{m}{2}*\bra{f\1_{ (\RRd\setminus B_{3R})\cup B_{R/2} }} }_{L^2_\unif}+ \norm{Y_m*\bra{f \1_{ B_{3R}\setminus B_{R/2}}}}_{L^2_\unif}\\
&\leq C\bra{e^{-\frac{mR}{4}} \norm{ f}_{L^2_\unif}+\norm{\1_{B_{3R}\setminus B_{R/2}}f}_{L^2_\unif}}.
\end{align*} 
As $\nabla Y_m$ is also exponentially decaying and is in $\ell^1(L^1)$, we proceed similarly for the second term of the RHS of~\eqref{eq:com_V_unif}. Finally we obtain the stated inequality
\begin{equation}\label{eq:VLinfini}
\norm{\eta_R Y_m*f-Y_m*(\eta_R f)}_{H^2_\unif}\leq \frac{C}{R}\bra{e^{-C'R}\norm{f}_{L^2_\unif}+\norm{ \1_{B_{3R}\setminus B_{R/2}} f}_{L^2_\unif}}.
\end{equation}

We turn to estimating $\norm{\com{\eta_R,\L} f}_{L^2_\unif}$. By~\eqref{eq:decomposition_com}, we have that 
\begin{align*}
 \com{\eta_R,\L} f=r_1+r_2= r_1+r_{21}+r_{22}
\end{align*}
where $r_1$ and $r_2$ are the densities associated with the first and seconds term of~\eqref{eq:decomposition_com} respectively, which are now locally trace class operators, and  $r_{21}$ is the density associated with the operator
\begin{align*}
\frac{1}{2i\pi}\oint_\C \com{\eta_R,\frac{1}{z-H_0}}Y_m*\bra{\1_{B_{3R}\setminus B_{R/2}}f}\frac{1}{z-H_0}dz.
\end{align*}
%
%
%
By~\eqref{eq:rho_Q1fcontroleeparYf_B} and using that, in dimension $d\leq 3$, $H^2_\unif(\RRd)\hookrightarrow L^\ii(\RRd)$, we find
\begin{align*}
 \norm{r_1}_{L^2_\unif}&\leq C \norm{Y_m*\bra{\eta_R f}-\eta_RY_m*f}_{L^\ii}\leq C  \norm{Y_m*\bra{\eta_R f}-\eta_RY_m*f}_{H^2_\unif}\\
&\leq \frac{C}{R}\bra{e^{-C'R}\norm{f}_{L^2_\unif}+\norm{ \1_{B_{3R}\setminus B_{R/2}} f}_{L^2_\unif}},
\end{align*}
where we have used~\eqref{eq:VLinfini} in the last step. Similarly for $r_{21}$, since $\norm{\Delta\eta_R}_{L^\ii}+\norm{\nabla\eta_R}_{L^\ii}\leq C/R$, we have
\begin{align}\label{eq:r_21} 
\norm{r_{21}}_{L^2_\unif}
&\leq \norm{\com{\eta_R, \Delta}\frac{1}{z-H_0}Y_m*( \1_{B_{3R}\setminus B_{R/2}}f)}_{\B } \nonumber\\
&\leq \norm{ \bra{(\Delta \eta_R)+ 2(\nabla\eta_R)\cdot \nabla } \frac{1}{z-H_0}}_{\B}\norm{Y_m*\bra{\1_{B_{3R}\setminus B_{R/2}}f}}_{L^\ii}\nonumber\\
&\leq \frac{C}{R}\norm{\1_{B_{3R}\setminus B_{R/2}}f}_{L^2_\unif}.
\end{align}
As to $r_{22}$, it is actually in $L^2(\RRd)$ and
\begin{align}\label{eq:r_22_unif}
 \norm{r_{22}}_{L^2_\unif}\leq \norm{r_{22}}_{L^2}&\leq \norm{\bra{-\Delta+1}^{-\frac12}\com{\eta_R, \Delta}\frac{1}{z-H_0}Y_m*( \1_{(\RRd\setminus B_{3R})\cup B_{R/2}}f)}_{\S_2}\nonumber\\
&\leq \frac{C}{R}\bra{e^{-C'R} \norm{f}_{L^2_\unif}+\norm{\1_{B_{3R}\setminus B_{R/2}}f}_{L^2_\unif}}.
\end{align}
The proof of~\eqref{eq:r_22_unif} is exactly the same than the proof of~\eqref{eq:r_2_part2}, except that in~\eqref{eq:yellow0}, we use the inequality $\norm{Y_m*f}_{L^\ii}\leq C\norm{f}_{L^2_\unif}$ instead of the inequality $\norm{Y_m*f}_{H^1}\leq C\norm{f}_{H^{-1}}$. 
This concludes the proof of the lemma.
\end{proof}

We pass now to the proof of Lemma~\ref{lemma:type_gronwall}.

\begin{proof}[Proof of Lemma~\ref{lemma:type_gronwall}]

We denote by $y_n=x_{\alpha ^n}$ and $b_n={C}{\alpha ^{-n}}e^{-C'\alpha ^n}$ for $n\in\NN$ and $\alpha \geq \alpha_0=\max\set{a, 2}$. By the assumption~\eqref{eq:assy_type_gron}, $x_{\alpha ^n/a}\leq x_{\alpha ^{n-1}}=y_{n-1}$, and we have
$$
y_n\leq b_n x_0 +\frac{C}{\alpha ^n}x_{\alpha ^n/a}\leq b_n x_0 +\frac{C}{\alpha ^n}y_{n-1}.
$$ 
Besides, there exists a continuous function $C(\alpha)$ such that $y_n\leq C(\alpha) z_n$, where $z_n={C^n}/{\alpha ^{n(n+1)/2}}z_0$ is a sequence defined by the induction relation
$ z_n={C}/{\alpha ^n}z_{n-1}$. Going back to $x_R$, we deduce that for any $n\in \NN\setminus{0}$ and $R=\alpha^n$, we have
\begin{align}\label{eq:decroissance_logarithmique_x_n}
 x_R &\leq C(\alpha) e^{-C'\frac{\log(R)^2}{\log(\alpha)}+C''\frac{\log(R)}{\log(\alpha)}}
x_0\leq C(\alpha)e^{-C'\frac{\log(R)^2}{\log(\alpha)}}x_0.
\end{align}
As~\eqref{eq:decroissance_logarithmique_x_n} holds true for any $\alpha\in\com{\alpha_0, \alpha_0^2}$, we deduce that there exists $C\geq 0$ independent of $\alpha$, but depending in general on $a$, such that for any $R\geq 2$,  
\begin{align*}
 x_R &\leq  Ce^{-C'{\log(R)^2}}x_0,
\end{align*}
which concludes the proof of the lemma.

\end{proof}


\section{Proof of Theorem~\ref{th:point_fixe} (Existence of ground states)}\label{sec:existence_GS}

Let us now establish the existence of a ground state for the perturbed crystal in the rHF framework. The proof of Theorem~\ref{th:point_fixe} is a consequence of our results on the operator $\L$ stated in the last section, and of the properties of the higher-order term in the expansion of $Q_f$ for a charge distribution $f\in L^2_\unif(\RRd)$.

To solve the self-consistent equation~\eqref{eq:SCE}, we first formulate the system in terms of the response electronic density $\rho=\rho_\gamma-\rho_{\gamma_0}$ as follow 
\begin{align}\label{eq:SCF_density}
\left\{\begin{array}{l}
       \rho =\rho_{Q}\\[0,2cm]
\displaystyle
 Q=\1_{H_0+V_\nu\leq 0}-\1_{H_0\leq 0} \\[0,2cm]
\displaystyle
-\Delta V_\nu+m^2V_\nu=\av{S^{d-1}}\bra{\rho-\nu}.
       \end{array}
\right.
\end{align}
Indeed, if ${\rho}$ is solution of~\eqref{eq:SCF_density}, then $\gamma=\1\bra{H_0+Y_m*({\rho}-\rho_{\gamma_0}-\nu)\leq 0}$ solves~\eqref{eq:SCE}. 
For a charge density $f\in L^2_\unif(\RRd)$, we expand
\begin{align*}
 Q_f=\1\bra{H_0+Y_m*f\leq 0}-\1\bra{H_0\leq 0}
\end{align*}
as powers of $f$ when $f$ is small. For this purpose, we assume that
\begin{align*}
 {\rm d}(\C,\sigma(H_0))\geq g,
\end{align*}
where $g={\rm d}(0,\sigma(H_0))$ and $\C$ is now a smooth curve in the complex plane enclosing the whole spectrum of $H_0$ below $0$ and crossing the real line at $0$ and at some point $c<\inf \sigma(H_0)-g$ (see Figure~\ref{fig:gap}).
Let us recall that for $V\in L^\ii(\RRd)$,  $\sigma\bra{H_0+V}\subset \sigma\bra{H_0}+\com{-\norm{V}_{L^\ii},\norm{V}_{L^\ii}}$. Therefore if $\norm{V}_{L^\ii}<g,$ then $H_0+V$ has a gap around $0$ and  $\sigma\bra{H}\subset [\inf \sigma\bra{H_0} -g,+\ii)$. 
For such a $V$, we have using Cauchy's residue formula, 
\begin{align*}
Q=\1\bra{H_0+V\leq 0}-\1\bra{H_0\leq 0} =\frac{1}{2i\pi}\oint_\C \frac{1}{z-H_0-V}dz-\frac{1}{2i\pi}\oint_\C \frac{1}{z-H_0} dz.
\end{align*}
By the resolvent formula, we obtain
\begin{align*}
Q
&=\frac{1}{2i\pi}\oint_\C \frac{1}{z-H_0} V\frac{1}{z-H_0}dz+ \frac{1}{2i\pi}\oint_\C \bra{\frac{1}{z-H_0} V}^2\frac{1}{z-H_0-V}dz.
\end{align*}
Therefore for $f\in L^2_\unif(\RRd)$ such that $\norm{f*Y_m}_{L^\ii}<g$,
\begin{align}\label{eq:Q=Q_1+Q_2tilde}
 Q_f=Q_{1,f}+\widetilde{Q}_{2,f},
\end{align}
where $Q_{1,f}$ has been defined and studied in Section~\ref{sec:linear_response} and $\widetilde{Q}_{2,f}$ is defined by
$$
\widetilde{Q}_{2,f}=\frac{1}{2i\pi}\oint_\C \bra{\frac{1}{z-H_0}Y_m*f}^2\frac{1}{z-H_0-Y_m*f}\,dz.
$$ 
We give some properties of the second order term $\widetilde{Q}_{2,f}$ in Lemma~\ref{lemma:LetL2continus} below.  Using the decomposition~\eqref{eq:Q=Q_1+Q_2tilde}, equation~\eqref{eq:SCF_density} becomes
\begin{align}\label{eq:SCE_density}
\rho&=\rho_{Q_{1, \rho-\nu}}+\rho_{\widetilde{Q}_{2,\rho-\nu}}=-\L(\rho-\nu)+ \rho_{\widetilde{Q}_{2,\rho-\nu}}. 
\end{align}
Following ideas of~\cite{HLS-05}, we recast~\eqref{eq:SCE_density} as 
\begin{align}\label{eq:rho}
 \rho=\frac{\L}{1+\L}\nu+\frac{1}{1+\L}\rho_{\widetilde{Q}_2(\rho-\nu)}.
\end{align}
In Proposition~\ref{prop:G_contractante} below, we show that for $\nu$ small enough, the operator $\G_\nu: \rho\mapsto {\L}\bra{1+\L}^{-1}\nu+\bra{1+\L}^{-1}\rho_{\widetilde{Q}_2(\rho-\nu)}$ admits a fixed point, which is controlled in the $L^2_\unif$ norm by the nuclear perturbation $\nu$. This will conclude the proof of Theorem~\ref{th:point_fixe}.

\begin{lemma}[Properties of the second order term]\label{lemma:LetL2continus} 
There exists $\delta_c>0$ and $C\geq 0$ such that for any $f\in L^2_\unif(\RRd)$ satisfying $\norm{f}_{L^2_\unif}\leq  \delta_c$, the operator $ \widetilde{Q}_{2,f}$ is trace class, the density $\rho_{\widetilde{Q}_{2,f}}$ is in $L^2_\unif(\RRd)$ and
$$
\norm{\rho_{\widetilde{Q}_{2,f}}}_{L^2_\unif}\leq C\norm{f}_{L^2_\unif}^2.
$$
\end{lemma}

\begin{proof}
Since $\norm{Y_m*f}_{L^\ii}\leq C_0\norm{f}_{L^2_\unif}$ (see~\eqref{eq:Young}), we can choose $\delta_c= g/2C_0$, where, we recall that $g={\rm d}(0,\sigma(H_0))$. 
In this case, $(z-H-Y_m*f)^{-1}(-\Delta+1)$ and its inverse are uniformly bounded w.r.t $z\in\C$ (see Lemma~\ref{lemma:B(z)borné}). Using the exact same procedure as in the proof of~\eqref{eq:rho_Q1fcontroleeparYf_B}, we obtain
that $\widetilde{Q}_{2,f}$ is trace class, $\rho_{\widetilde{Q}_{2,f}}\in L^2_\unif(\RRd)$ and 
\begin{align*}
 \norm{\rho_{\widetilde{Q}_{2,f}}}_{L^2_\unif}&\leq C \norm{\oint_\C Y_m*f\frac{1}{z-H_0}Y_m*f\,dz}_{\B}\leq C\norm{Y_m*f}_{L^\ii}^2 \leq C\norm{f}_{L^2_\unif}^2,
\end{align*}
which concludes the proof of the lemma.
\end{proof}

\begin{proposition}\label{prop:G_contractante}
There exists $\alpha_c, \epsilon>0$ such that if $\norm{\nu}_{L^2_\unif}\leq \alpha_c$, then 
$$
\begin{array}{lrll}
 \G_\nu: & B_{L^2_\unif}(\epsilon)& \rightarrow &B_{L^2_\unif}(\epsilon)\\
 & \rho&\mapsto &  \frac{\L}{1+\L}\nu+\frac{1}{1+\L}\rho_{\widetilde{Q}_{2,\rho-\nu}}
\end{array}
$$
is well-defined and contracting on $B_{L^2_\unif}(\epsilon)=\set{f\in L^2_\unif(\RRd),\; \norm{f}_{L^2_\unif}\leq \epsilon}$.
Thus, it admits a unique fixed point $\rho$ in the ball $B_{L^2_\unif}(\epsilon)$. Moreover $\rho$ satisfies
\begin{equation}\label{eq:rho_controlee_par_nu}
 \norm{\rho}_{L^2_\unif}\leq C \norm{\nu}_{L^2_\unif},
\end{equation}
for a constant $C$ independent of $\nu$. 
\end{proposition}

\begin{proof}
We want to use Lemma~\ref{lemma:LetL2continus} to show that $\G$ is well-defined on a small ball of $L^2_\unif(\RRd)$. Here, the  charge distribution is $f=\rho- \nu$. We thus need to choose $\alpha_c$ and $\epsilon$  such that 
$\norm{\rho-\nu}_{L^2_\unif}\leq \norm{\rho}_{L^2_\unif}+\norm{\nu}_{L^2_\unif}\leq \epsilon+ \alpha_c\leq\delta_c,$ 
where $\delta_c$ is given by Lemma~\ref{lemma:LetL2continus}. Let $A>0$, $0 <\epsilon\leq \delta_c/(1+A)$ and $\alpha_c=A\epsilon$. Let $\nu$ and $\rho$ such that $\norm{\nu}_{L^2_\unif}\leq \alpha_c$ and $\norm{\rho}_{L^2_\unif}\leq \epsilon$. By Lemma~\ref{lemma:LetL2continus} and the fact that $\L$ and $1/(1+\L)$ are bounded on $L^2_\unif(\RRd)$ (see Theorem~\ref{th:props_de_L}), we have
\begin{align}\label{eq:rho_controlée_par_nu}
 \norm{\G_\nu(\rho)}_{L^2_\unif}&\leq \norm{\frac{\L}{1+\L}}_{\B(L^2_\unif)}\norm{\nu}_{L^2_\unif}+ \norm{\frac{1}{1+\L}}_{\B(L^2_\unif)}\norm{\rho_{\widetilde{Q}_{2,\rho-\nu}}}_{L^2_\unif}\nonumber\\
&\leq C_1\norm{\nu}_{L^2_\unif}+C_2\norm{\rho-\nu}_{L^2_\unif}^2\leq \bra{C_1A+ C_2(1+A)^2\epsilon}\epsilon.
\end{align}
We choose $A<{1}/{C_1}$ such that for $\epsilon\leq (1-AC_1)/(C_2(1+A)^2)$, we have
\begin{align*}
\norm{\G_\nu(\rho)}_{L^2_\unif}&\leq \epsilon.
\end{align*}
To show that $\G_\nu$ is contracting on $B_{L^2_\unif}(\epsilon)$ for $\epsilon$ small enough, we use the explicit expression of $\widetilde{Q}_{2,\rho-\nu}$. Let $\rho,\rho' \in B_{L^2_\unif}(\epsilon)$ and denote by $H=H_0+Y_m*(\rho-\nu)$ and $H'=H_0+Y_m*(\rho'-\nu)$. The function $(1+\L)\bra{\G_\nu(\rho)-\G_\nu(\rho') }$ is the density associated with the operator 
\begin{align*}
&\frac{1}{2i\pi}\oint_\C\bra{\frac{1}{z-H_0}Y_m*(\rho-\nu) }^2\frac{1}{z-H}-\bra{\frac{1}{z-H_0}Y_m*(\rho'-\nu)}^2\frac{1}{z-H'} \,dz.
\end{align*}
A straightforward calculation shows that this operator can be written as
\begin{align}\label{eq:decomposition_Q2tilde}
&\frac{1}{2i\pi}\oint_\C\bra{\frac{1}{z-H_0}Y_m*(\rho-\nu)} ^2\frac{1}{z-H}Y_m*(\rho-\rho')\frac{1}{z-H'}\nonumber\\
&\qquad \qquad\qquad+\frac{1}{z-H_0}Y_m*(\rho-\nu) \frac{1}{z-H_0}Y_m*(\rho-\rho')\frac{1}{z-H'}\nonumber\\
&\qquad \qquad\qquad+\frac{1}{z-H_0}Y_m*(\rho-\rho') \frac{1}{z-H_0}Y_m*(\rho'-\nu)\frac{1}{z-H'}dz.
\end{align}
Using the same techniques as before, we deduce that 
\begin{align*}
 \norm{\G_\nu(\rho)-\G_\nu(\rho') }_{L^2_\unif}
&\leq C_3\bra{ \norm{\rho}_{L^2_\unif}+ \norm{\rho'}_{L^2_\unif}+ \norm{\nu}_{L^2_\unif}  }\norm{\rho-\rho'}_{L^2_\unif} \\
&\leq C_3\bra{2+A}\epsilon \norm{\rho-\rho'}_{L^2_\unif}.
\end{align*}
Taking, in addition, $\epsilon< 1/(C_3(2+A))$, we have that $\G_\nu$ is contracting on $B_{L^2_\unif}(\epsilon)$. Let $\rho$ be the unique fixed point of $\G_\nu$ in $B_{L^2_\unif}(\epsilon)$. It remains to prove~\eqref{eq:rho_controlee_par_nu}. By~\eqref{eq:rho_controlée_par_nu}, we have
\begin{align*}
 \norm{\rho}_{L^2_\unif}=\norm{\G_\nu(\rho)}_{L^2_\unif}\leq C_1\norm{\nu}_{L^2_\unif}+ C_2(1+A)\epsilon \bra{\norm{\rho}_{L^2_\unif}+\norm{\nu}_{L^2_\unif} }.
\end{align*}
Therefore 
$\bra{1- C_2(1+A)\epsilon} \norm{\rho}_{L^2_\unif}\leq \bra{C_1+  C_2(1+A)\epsilon}\norm{\nu}_{L^2_\unif}$.
Using that $\epsilon\leq \bra{1-AC_1}/\bra{C_2(1+A)^2}$, we have $1- C_2(1+A)\epsilon>0$ and we deduce that 
\begin{align*}
 \norm{\rho}_{L^2_\unif}\leq \frac{C_1+  C_2(1+A)\epsilon}{1- C_2(1+A)\epsilon}\norm{\nu}_{L^2_\unif}\leq \frac{1}{A} \norm{\nu}_{L^2_\unif},
\end{align*}
which concludes the proof of the proposition. 
\end{proof}


\section{Proofs of Theorem~\ref{th:decay} and Proposition~\ref{prop:loc} (Decay estimates)}\label{sec:decay_rate}

We present in this section the proofs of Theorem~\ref{th:decay} and Proposition~\ref{prop:loc}. They consist in decay estimates of the mean-field potential $V_\nu$ and the mean-field density $\rho_\nu$. These estimates are used later on in the proofs of Theorems~\ref{th:limite_thermo} and~\ref{th:expansion}. 

\subsection{Proof of Theorem~\ref{th:decay}}

\begin{proof}[Proof of Theorem~\ref{th:decay}]
Assume that $\norm{\nu}_{L^2_\unif}\leq \alpha_c$, where $\alpha_c$ is given in Theorem~\ref{th:point_fixe}. We use the notation $\rho$ to denote the mean-field density $\rho_\nu=\rho_{\gamma_\nu-\gamma_0}$, the solution of~\eqref{eq:rho}, and denote by $V=V_\nu=Y_m*(\rho-\nu)$.
Recall the decomposition~\eqref{eq:SCE_density} of $\rho$ in a linear term and a higher order term 
\begin{align*}
 \rho=-\L\bra{\rho-\nu}+ \rho_{\widetilde{Q}_{2,\rho-\nu}}.
\end{align*}
Using localizing functions, we will show that $\rho$ decays far from the support of $\nu$. To do so, let us introduce the set $I=\set{k\in \ZZd,\; \text{supp}(\nu)\cap B(0,1)+k\neq \emptyset }$ and for $R\geq 1$, the set $B_R=B_{I,R}
$ and the the function $\chi_R=\chi_{I,R}$ defined in Lemma~\ref{lemma:estimee_com_L2unif}.
 We denote by $\eta_R=1-\chi_R$.
We thus have
\begin{align*}
\eta_R\rho&=-\eta_R\L(\rho-\nu)+\eta_R\rho_{\widetilde{Q}_{2,\rho-\nu}}=  -\L\eta_R(\rho-\nu)+\com{\L,\eta_R}(\rho-\nu)+\eta_R\rho_{\widetilde{Q}_{2,\rho-\nu}}.
\end{align*}
As for $R\geq 1$, $\eta_R\nu=0$, it follows 
\begin{align}\label{eq:diviser_eta_RrhoQ}
\eta_R\rho&= \frac{1}{\bra{1+\L}}\com{\L,\eta_R}(\rho-\nu)+\frac{1}{\bra{1+\L}}\eta_R\rho_{\widetilde{Q}_{2,\rho-\nu}}.
\end{align} 
We will successively bound each term of the RHS of~\eqref{eq:diviser_eta_RrhoQ}. For the first term, we have by Lemma~\ref{lemma:estimee_com_L2unif} for $R\geq 2$,  
\begin{align}\label{eq:borne_commutateur}
 \norm{\frac{1}{\bra{1+\L}}\com{\L,\eta_R}(\rho-\nu)  }_{L^2_\unif}&\leq \frac{C}{R}\bra{e^{-C'R}\norm{\rho-\nu}_{L^2_\unif}+ \norm{\1_{B_{3R}\setminus B_{R/2}}\bra{\rho-\nu}}_{L^2_\unif}}\nonumber\\
&\leq \frac{C}{R}\bra{e^{-C'R}\norm{\nu}_{L^2_\unif}+ \norm{\1_{B_{3R}\setminus B_{R/2}}\rho}_{L^2_\unif}},
\end{align}
where we have used that $\1_{B_{3R}\setminus B_{R/2}}\nu=0$ for $R\geq 2$, that $\rho$ is controlled by $\nu$ in the $L^2_\unif$ norm and that $1/(1+\L)$ is bounded on $L^2_\unif(\RRd)$.
As to the second term of the RHS of~\eqref{eq:diviser_eta_RrhoQ}, since $\1_{\RRd\setminus B_{R}}\eta_R=\eta_R$, we have
\begin{align}\label{eq:division_rho2_tilde}
\eta_R\widetilde{Q}_{2,\rho-\nu}
&=\frac{1}{2i\pi}\oint_\C\frac{1}{z-H_0}\1_{\RRd\setminus B_{R}}V \frac{1}{z-H_0} \eta_{R} V \frac{1}{z-H}dz\nonumber\\
&\quad+\frac{1}{2i\pi}\oint_\C\frac{1}{z-H_0}\1_{\RRd\setminus B_{R}}V\com{\eta_{R}, \frac{1}{z-H_0}} V \frac{1}{z-H}dz\nonumber\\
&\quad+ \frac{1}{2i\pi}\oint_\C\com{\eta_R,\frac{1}{z-H_0}} V \frac{1}{z-H_0} V \frac{1}{z-H}dz,
\end{align}
where $H=H_0+V$ and $\C$ is as in the previous section. We recall that by the assumption $\norm{\nu}_{L^2_\unif}\leq \alpha_c$, the operator $H$ has a gap around $0$, thus the operator $(z-H)^{-1}(-\Delta+1)$ and its inverse are uniformly bounded on $\C$ and all the estimates obtained in the previous sections hold when we replace $H_0$ by $H$. We denote by $r_3$, $r_4$ and $r_5$ the densities associated with the three operators of the RHS of~\eqref{eq:division_rho2_tilde} respectively. 
Using an inequality similar to~\eqref{eq:rho_Q1fcontroleeparYf_B}, involving $H$ instead of $H_0$ in the resolvent in the right, we have
\begin{align*}
 \norm{r_3}_{L^2_\unif}&\leq C\int_\C\norm{\1_{\RRd\setminus B_{R}}V\frac{1}{z-H_0}V\eta_{R}}_{\B}dz \leq C\norm{V\1_{\RRd\setminus B_{R}}}_{L^\ii}\norm{V\eta_{R}}_{L^\ii}. 
\end{align*}
By~\eqref{eq:estimee_com_L2unif} in Lemma~\ref{lemma:estimee_com_L2unif}, and using that $\norm{Y_m*f}_{H^2_\unif}=\norm{f}_{L^2_\unif}$, we have that for $R\geq 2$
\begin{align}\label{eq:eta_R_V}
 \norm{\eta_RV}_{H^2_\unif}\leq \norm{\eta_R\rho}_{L^2_\unif}+ \frac{C}{R}\bra{e^{-C'R}\norm{\nu}_{L^2_\unif}+\norm{ \1_{ B_{3R}\setminus B_{R/2} }\rho}_{L^2_\unif}}.
\end{align}
Therefore 
\begin{align*}
 \norm{r_3}_{L^2_\unif}
&\leq  C\norm{\1_{\RRd\setminus B_{R}}V}_{L^\ii} \bra{C\norm{ \eta_R \rho}_{L^2_\unif}+ \frac{C}{R}e^{-C'R}\norm{\nu}_{L^2_\unif}+
\frac{C}{R}\norm{ \1_{ B_{3R}\setminus B_{R/2} } \rho}_{L^2_\unif}}.
\end{align*}
To bound $r_4$ and $r_5$, we recall that we have shown in the proof of~\eqref{eq:estimee_com_L2unif} (see~\eqref{eq:r_21} and~\eqref{eq:r_22_unif}) that for any $f\in L^2_\unif(\RRd)$
\begin{align*}
&\norm{\bra{-\Delta+1}^{-\frac12}\com{\eta_R,\Delta}\frac{1}{z-H_0}Y_m*\bra{\1_{(\RRd\setminus B_{3R})\cup B_{2R}}f}}_{\S_2}\\
&\quad\quad+\norm{\com{\eta_R,\Delta}\frac{1}{z-H_0}Y_m*\bra{\1_{B_{3R}\setminus B_{2R}}f}}_{\B}\leq \frac{C}{R} \bra{ e^{-C'R}\norm{f}_{L^2_\unif}+\norm{\1_{ B_{3R}\setminus B_{R/2}}f}_{L^2_\unif} }.
\end{align*}
Therefore, using again the equality $\com{\eta_R,(z-H_0)^{-1}}=-(z-H_0)^{-1}\com{\eta_R,\Delta}(z-H_0)^{-1}$, and an inequality similar to~\eqref{eq:rho_Q1fcontroleeparYf_B}, we obtain that for any $R\geq 2$,
\begin{align}\label{eq:r4}
 \norm{r_4}_{L^2_\unif}&\leq C\oint_{\C}\norm{\1_{\RRd\setminus B_R}V}_{L^\ii}\norm{\frac{1}{z-H_0}\com{\eta_R,\Delta}\frac{1}{z-H_0}V}_{\B}dz\nonumber\\
&\leq \frac{C}{R} \norm{\1_{\RRd\setminus B_R}V}_{L^\ii}\bra{ e^{-C'R}\norm{\nu}_{L^2_\unif}+\norm{\1_{ B_{3R}\setminus B_{R/2}}\rho}_{L^2_\unif} }.
\end{align}
The last term of the RHS of~\eqref{eq:division_rho2_tilde} can be written $Q_{\rm in}+Q_{\rm out}$, where 
$$Q_{\rm in}=\frac{1}{2i\pi}\oint_\C\com{\eta_R,\frac{1}{z-H_0}} Y_m*\bra{\1_{B_{3R}\setminus B_{2R} }\bra{\rho-\nu}} \frac{1}{z-H_0} V \frac{1}{z-H}dz.$$
In the same way we obtained~\eqref{eq:r4}, we get 
\begin{align*}
 \norm{\rho_{Q_{\rm in}}}_{L^2_\unif}
&\leq \frac{C}{R}\norm{V}_{L^\ii} \bra{ e^{-C'R}\norm{\nu}_{L^2_\unif}+\norm{\1_{ B_{3R}\setminus B_{R/2}}\rho}_{L^2_\unif} }.
\end{align*}
To estimate $\rho_{Q_{\rm out}}$, we recall that by~\eqref{eq:generic_A}, we have that for any $A\in \S_2(L^2(\RRd))$ 
\begin{align*}
\norm{ \rho_{\bra{-\Delta+1}^{-\frac{1}{2}}A\bra{-\Delta+1}^{-\frac{1}{2}}}}_{L^2}
\leq C\norm{A}_{\S_2}. 
\end{align*}
Therefore 
\begin{align*}
 \norm{\rho_{Q_{\rm out}}}_{L^2_\unif}&\leq C\oint_{\C}  \norm{ (-\Delta+1)^{-\frac12}\com{\eta_R,\Delta} \frac{1}{z-H_0} Y_m*\bra{\1_{(\RRd\setminus B_{3R})\cup B_{2R} }\bra{\rho-\nu}}  }_{\S_2}\nonumber\\
& \qquad\times\norm{\frac{1}{z-H_0}V(1-\Delta)^{-\frac12}}_{\B} dz\nonumber\\
&\leq \frac{C}{R}\norm{V}_{L^\ii} \bra{ e^{-C'R}\norm{\nu}_{L^2_\unif}+\norm{\1_{ B_{3R}\setminus B_{R/2}}\rho}_{L^2_\unif} }.
\end{align*}
Now that we have found estimates on $r_3$, $r_4$ and $r_5=\rho_{Q_{\rm in}}+\rho_{Q_{\rm out}}$, we use that 
\begin{align*}
\norm{\1_{\RRd\setminus B_{R}}V}_{L^\ii}\leq \norm{V}_{L^\ii}\leq C\norm{\rho-\nu}_{L^2_\unif}\leq C\norm{\nu}_{L^2_\unif}\leq C\alpha_c,
\end{align*}
to estimate $\eta_R\rho_{\widetilde{Q}_{2,\rho-\nu}}$ as follow
\begin{align}\label{eq:etaRrhoQ2tilde}
 \norm{\eta_R\rho_{\widetilde{Q}_{2,\rho-\nu}}}_{L^2_\unif}&\leq C\alpha_c\norm{ \eta_R \rho}_{L^2_\unif}+ \frac{C}{R}e^{-C'R}\norm{\nu}_{L^2_\unif}+\frac{C}{R}\norm{ \1_{ B_{3R}\setminus B_{R/2} } \rho}_{L^2_\unif}.
\end{align}
Using once more that $1/(1+\L)$ is bounded on $L^2_\unif(\RRd)$, we deduce 
in view of~\eqref{eq:diviser_eta_RrhoQ},~\eqref{eq:borne_commutateur} and~\eqref{eq:etaRrhoQ2tilde}
\begin{align*}
 \norm{\eta_R\rho}_{L^2_\unif}&\leq C_0\alpha_c\norm{ \eta_R \rho}_{L^2_\unif}+ \frac{C}{R}e^{-C'R}\norm{\nu}_{L^2_\unif}+\frac{C}{R}\norm{ \1_{ B_{3R}\setminus B_{R/2} } \rho}_{L^2_\unif}.
\end{align*}
We choose $\alpha_c'\leq\min\set{ 1/(2C_0),\alpha_c}$, 
and assume that $\norm{\nu}_{L^2_\unif}\leq \alpha_c'$. It follows 
\begin{align*}
 \norm{\eta_R\rho}_{L^2_\unif}\leq  \frac{C}{R}e^{-C'R}\norm{\nu}_{L^2_\unif}+\frac{C}{R}\norm{ \1_{ B_{3R}\setminus B_{R/2} } \rho}_{L^2_\unif}.
\end{align*}
We have a similar inequality for $V$. Indeed, by~\eqref{eq:eta_R_V}, we have
\begin{align}\label{eq:recurence_V}
 \norm{\eta_RV}_{H^2_\unif}
&\leq \norm{\eta_R\rho}_{L^2_\unif}+ \frac{C}{R}e^{-C'R}\norm{\nu}_{L^2_\unif}+\frac{C}{R}\norm{ \1_{ B_{3R}\setminus B_{R/2} }\rho}_{L^2_\unif}\nonumber\\
&\leq\frac{C}{R}e^{-C'R}\norm{\nu}_{L^2_\unif}+\frac{C}{R}\norm{ \1_{ B_{3R}\setminus B_{R/2} } \rho}_{L^2_\unif}.
\end{align}
Using Lemma~\ref{lemma:type_gronwall} with $x_R$ to  $\norm{\1_{\RRd\setminus B_R}\rho}_{L^2_\unif}$, we obtain
\begin{align}\label{eq:decay_rho}
\norm{\eta_R\rho}_{L^2_\unif}\leq C e^{-C'\bra{\log R}^2}\norm{\nu}_{L^{2}_\unif}.
\end{align}
Inserting~\eqref{eq:decay_rho} in~\eqref{eq:recurence_V}, we get
\begin{align*}
  \norm{\eta_R V}_{H^2_\unif}\leq C e^{-C'\bra{\log R}^2}\norm{\nu}_{L^{2}_\unif}.
\end{align*}
Finally, noticing that $\1_{\RRd\setminus C_R(\nu)}\leq \eta_{R/2}$, we conclude the proof of~\eqref{eq:decayL2unif}.

\end{proof}

We now turn to the


\subsection{Proof of Proposition~\ref{prop:loc}}

\begin{proof}[Proof of Proposition~\ref{prop:loc}]
Assume that $\norm{\nu}_{L^2_\unif}\leq \alpha_c$, where $\alpha_c$ is given in Theorem~\ref{th:point_fixe}. 
As $\rho_\nu$ and $\rho_{\nu_L}$ are fixed points of the functionals $\G_{\nu}$ and $\G_{\nu_L}$ respectively, then 
\begin{align*}
 \rho_\nu -\rho_{\nu_L}&=\frac{\L}{1+\L}(\nu-\nu_L)+\frac{1}{1+\L}\rho_{\widetilde{Q_2}(\rho_\nu-\nu)-\widetilde{Q_2}(\rho_{\nu_L}-\nu_L)}.
\end{align*}
For $R\geq 1$, let $\chi_{R}=\chi_{\set{0},R}$ and $B_R=B_{\set{0},R}$ as defined in Lemma~\ref{lemma:estimee_com_L2unif}. Since $\1_{B_{R}}\leq \chi_{R}$, then
\begin{align*}
 \norm{\1_{B_{R}}  \bra{\rho_\nu -\rho_{\nu_L}}}_{L^2_\unif}&\leq  \norm{ \chi_{R} \bra{\rho_\nu -\rho_{\nu_L}}}_{L^2_\unif}\\
&\leq  \norm{ \chi_{R}\frac{1}{1+\L} \bra{\L\bra{\nu-\nu_L}+ \rho_{\widetilde{Q_2}(\rho_\nu-\nu)-\widetilde{Q_2}(\rho_{\nu_L}-\nu_L)} }}_{L^2_\unif}.
\end{align*}
Besides, there exists $C\geq 0$ such that for any $f\in L^2_\unif(\RRd)$ and any $R\geq 1$,
\begin{align}\label{eq:estimee_com_L2unif1/1+L}
\norm{\com{\chi_R,\frac{1}{1+\L}}f}_{L^2_\unif}\leq \frac{C}{R}\bra{ e^{-C'R}\norm{f }_{L^2_\unif}+\norm{\1_{B_{3R}}\frac{1}{1+\L}f }_{L^2_\unif} }.
\end{align} 
 Indeed, using that $1/(1+\L)$ is bounded on $L^2_\unif(\RRd)$ and estimate~\eqref{eq:estimee_com_L2unif} in Lemma~\ref{lemma:estimee_com_L2unif} (notice that $\L\chi_R-\chi_R\L=\eta_R\L-\L\eta_R$), we obtain
\begin{align*}
 \norm{\com{\chi_R,\frac{1}{1+\L}}f}_{L^2_\unif}&=\norm{\frac{1}{1+\L}\com{ \L,\chi_R }\frac{1}{1+\L}f }_{L^2_\unif}\\
&\leq C\norm{\com{ \L,\chi_R }\frac{1}{1+\L}f }_{L^2_\unif}\\
&\leq \frac{C}{R}\bra{ e^{-C'R}\norm{\frac{1}{1+\L}f }_{L^2_\unif}+\norm{\1_{B_{3R}\setminus B_{R/2}}\frac{1}{1+\L}f }_{L^2_\unif} }\\
&\leq \frac{C}{R}\bra{ e^{-C'R}\norm{f }_{L^2_\unif}+\norm{\1_{B_{3R}\setminus B_{R/2}}\frac{1}{1+\L}f }_{L^2_\unif} }.
\end{align*}
Using~\eqref{eq:estimee_com_L2unif1/1+L} for $f=\L\bra{\nu-\nu_L}+ \rho_{\widetilde{Q_2}(\rho_\nu-\nu)-\widetilde{Q_2}(\rho_{\nu_L}-\nu_L)}$, we have
\begin{align}\label{eq:chi1/1+Lr_el}
\norm{\chi_{R}\bra{\rho_{\nu}-\rho_{\nu_L}}}_{L^2_\unif}
&\leq \norm{ \frac{1}{1+\L}\chi_{R}\L\bra{\nu-\nu_L}}_{L^2_\unif}
+\norm{ \frac{1}{1+\L}\chi_{R}\rho_{\widetilde{Q_2}(\rho_\nu-\nu)-\widetilde{Q_2}(\rho_{\nu_L}-\nu_L)}}_{L^2_\unif}\nonumber\\
&\quad+\frac{C}{R}e^{-C'R}\norm{\nu}_{L^2_\unif}+ \frac{C}{R}\norm{\1_{B_{3R}} \bra{\rho_{\nu}-\rho_{\nu_L}} }_{L^2_\unif}.
\end{align}
We first bound the first term of the RHS of~\eqref{eq:chi1/1+Lr_el}. Using~\eqref{eq:estimee_com_L2unif} in Lemma~\ref{lemma:estimee_com_L2unif} and that for $R\leq L/4$ it holds that $\chi_{R}(\nu-\nu_L)=0$, we have for $R\leq L/4$ 
 \begin{align}\label{eq:L/L+1nu-nu_L}
&\norm{ \frac{1}{1+\L}\chi_{R}\L\bra{\nu-\nu_L}}_{L^2_\unif}\leq C \norm{\chi_{R}\L(\nu-\nu_L)}_{L^2_\unif}\nonumber\\
&\qquad\qquad\leq C\norm{\chi_{R}(\nu-\nu_L)}_{L^2_\unif}+ \frac{C}{R}e^{-C'R}\norm{\nu }_{L^2_\unif}+\frac{C}{R} \norm{\1_{B_{3R} }\bra {\nu-\nu_L} }_{L^2_\unif}\nonumber\\
&\qquad\qquad\leq \frac{C}{R}e^{-C'R}\norm{\nu }_{L^2_\unif}+ \frac{C}{R} \norm{\1_{B_{3R} }\bra {\nu-\nu_L} }_{L^2_\unif}.
 \end{align}
We turn to the second term of the RHS of~\eqref{eq:chi1/1+Lr_el}. Using~\eqref{eq:estimee_com_L2unif}, a decomposition similar to~\eqref{eq:decomposition_Q2tilde} and reasoning as in the proof of~\eqref{eq:etaRrhoQ2tilde}, we find for $R\leq L/4$
\begin{align*}
& \norm{ \frac{1}{1+\L}\chi_{R}\rho_{\widetilde{Q_2}(\rho_\nu-\nu)-\widetilde{Q_2}(\rho_{\nu_L}-\nu_L)}}_{L^2_\unif}
\leq C_0 \bra{\norm{\nu}_{L^2_\unif}+\norm{\nu}_{L^2_\unif}^2}\\
&\qquad\qquad\qquad\times\left(\norm{\chi_{R}(\rho_{\nu}-\rho_{\nu_L}) }_{L^2_\unif}+  \frac{C}{R}  e^{-C'R}\norm{\nu}_{L^2_\unif}+ \frac{C}{R}\norm{\1_{B_{3R} }f_L}_{L^2_\unif} \right).
\end{align*}
We choose $\alpha_c'\leq\alpha_c$ such that $C_0(\alpha_c'+{\alpha_c'}^2)\leq 1/2 $. Thus, if $\norm{\nu}_{L^2_\unif}\leq \alpha_c'$ then 
\begin{align}\label{eq:rho_Q_2tilde}
 &\norm{\frac{1}{1+\L}\chi_{R}\rho_{\widetilde{Q_2}(\rho_\nu-\nu)-\widetilde{Q_2}(\rho_{\nu_L}-\nu_L)}}_{L^2_\unif}\leq \frac12 \norm{\chi_{R}(\rho_{\nu}-\rho_{\nu_L}) }_{L^2_\unif}\nonumber\\
&\qquad\qquad\qquad\qquad\qquad\qquad\qquad
+ \frac{C}{R} \bra{ e^{-C'R}\norm{\nu}_{L^2_\unif} + \norm{\1_{B_{3R} }f_L}_{L^2_\unif} }.
\end{align}
In this case, combining~\eqref{eq:chi1/1+Lr_el},~\eqref{eq:L/L+1nu-nu_L} and~\eqref{eq:rho_Q_2tilde}, we obtain for $R\leq L/4$
 \begin{align*}
&\norm{\1_{B_{R}}\bra{\rho_\nu-\rho_{\nu_L}}}_{L^2_\unif}\leq  \norm{\chi_{R}\bra{\rho_\nu-\rho_{\nu_L}}}_{L^2_\unif}\\
&\qquad\qquad\leq\frac{C}{R}\bra{ e^{-C'R}\norm{\nu}_{L^2_\unif}+ \norm{\1_{B_{3R} }\bra{\rho_\nu-\rho_{\nu_L}} }_{L^2_\unif}+\norm{\1_{B_{3R} }\bra{\nu-{\nu_L}} }_{L^2_\unif}}.
 \end{align*}
Using a recursion argument, we easily see that for any $\beta \geq 1$, there exists $C\geq 0$ such that 
\begin{align*}
\norm{\rho_\nu-\rho_{\nu_L}}_{L^2_\unif(B_{L/4^\beta})}&\leq   \frac{C}{L^\beta } e^{-C'L}\norm{\nu}_{L^2_\unif}+ \frac{C}{L^\beta }\norm{\1_{B_L }\bra{\rho_\nu-\rho_{\nu_L}} }_{L^2_\unif}\nonumber\\
& \quad+\frac{C}{L^\beta }\norm{\1_{B_L }\bra{\nu-{\nu_L}} }_{L^2_\unif} \leq \frac{C}{L^\beta }\norm{\nu}_{L^2_\unif}.
\end{align*}
To conclude the proof of the proposition, it remains to prove the bound on the potential. Using~\eqref{eq:estimee_com_L2unif} and denoting by $f_L=\rho_\nu-\rho_{\nu_L}-\nu+\nu_L$, we have 
\begin{align*}
& \norm{V_\nu-V_{\nu_L}}_{H^2_\unif(B_{L/{4^\beta}})}
\leq \norm{\chi_{L/{4^\beta}}Y_m*f_L}_{H^2_\unif}\\
&\qquad\qquad\leq C\norm{\chi_{L/{4^\beta}}(\rho_\nu-\rho_{\nu_L})}_{L^2_\unif}+\frac{C}{L}\bra{e^{-C'L}\norm{\nu}_{L^2_\unif}+ \norm{\1_{B_{3L/{4^\beta}} }f_L} _{L^2_\unif}}\\
&\qquad\qquad\leq \frac{C}{L^\beta}\norm{\nu}_{L^2_\unif}+\frac{C}{L}\bra{e^{-C'L}\norm{\nu}_{L^2_\unif}+ \frac{C}{L^{\beta-1}}\norm{\nu}_{L^2_\unif}}
\leq \frac{C}{L^\beta}\norm{\nu}_{L^2_\unif}.
\end{align*}

\end{proof}


\section{Proof of Theorem~\ref{th:limite_thermo} (Thermodynamic limit)}\label{sec:thermo_limit}
\begin{proof}[Proof of  Theorem~\ref{th:limite_thermo}]
Assume that $\norm{\nu}_{L^2_\unif}\leq \alpha_c$, where $\alpha_c$ is given by Proposition~\ref{prop:loc}. 
By Cauchy's formula, we have
$$
\gamma_\nu-\gamma_{\nu_L}=\frac{1}{2i\pi}\int_\C \frac{1}{z-H_0-V_\nu}-\frac{1}{z-H_0-V_{\nu_L}}dz,
$$
where the curve $\C$ is as in Section~\ref{sec:existence_GS}.
We write the resolvent difference as
\begin{align*}
&\frac{1}{z-H_0-V_\nu}-\frac{1}{z-H_0-V_{\nu_L}}=\frac{1}{z-H_0-V_\nu}Y_m*f_L\frac{1}{z-H_0-V_{\nu_L}},
\end{align*} 
where $f_L=\rho_\nu-\nu -\rho_{\nu_L}+ \nu_L$. For a compact set $B\subset \RRd$, we have  
\begin{align*}
&\tr\av{\1_B \frac{1}{z-H_0-V_\nu} Y_m*f_L\frac{1}{z-H_0-V_{\nu_L}}\1_B}\leq C \norm{ \1_B \frac{1}{z-H_0-V_\nu}Y_m*f_L}_{\S_2}.
\end{align*}
For $L$ large enough, we have $B\subset B(0,L/8)$ and, by Proposition~\ref{prop:loc},
\begin{align*}
 \norm{ \1_B \frac{1}{z-H_0-V_\nu}\1_{B(0,L/4)}Y_m*f_L}_{\S_2}&\leq\norm{\1_B \frac{1}{z-H_0-V_\nu} }_{\S_2} \norm{\1_{B(0,L/4)}Y_m*f_L}_{L^\ii} \\& \leq\frac{C}{L}\norm{\nu}_{L^2_\unif}.
\end{align*}
Besides, as ${\rm d}(B, \RRd\setminus B_{L/4})\geq L/8$, we have using Lemma~\ref{lemma:CT},
\begin{align*}
 \norm{ \1_B \frac{1}{z-H_0-V_\nu}\1_{\RRd\setminus B(0,L/4)}Y_m*f_L}_{\S_2}&\leq Ce^{-C'L}\norm{\1_{\RRd\setminus B(0,L/4)}Y_m*f_L}_{L^\ii}\\
& \leq \frac{C}{L}\norm{\nu}_{L^2_\unif}.
\end{align*} 
As $\C$ is a compact set and all the estimates are uniform on $\C$, we conclude that 
\begin{align*}
 \norm{\1_B\bra{\gamma_\nu-\gamma_{\nu_L}}\1_B}_{\S_1}\leq \frac{C}{L}\norm{\nu}_{L^2_\unif}\cvL 0.
\end{align*}

\end{proof}

\section{Proof of Theorem~\ref{th:expansion} (Expansion of the density of states)}\label{sec:expansion}

The proof of Theorem~\ref{th:expansion} follows essentially the proof of~\cite[Theorem 1.1]{Klopp-95}. The main difference is the proof of Proposition~\ref{prop:definition_termes_expansion+reste} below, which deals with self-consistent potentials, while~\cite[Proposition 2.1]{Klopp-95} deals with linear potentials. Treating nonlinear potentials is done at the price of assuming that the defect $\chi$ is small in the $L^2_\unif$-norm, so that the potential decays fast enough. 
For the sake of self-containment, we mention here the main steps of the proof; more details can be found in~\cite{these}.

\begin{proof}[Proof of Theorem~\ref{th:expansion} ]

Following~\cite{Klopp-95}, we first express the density of states  of the random operator $H_p(\o)$ in terms of the resolvent $(z-H_p)^{-1}$ for $z\in\CC$. We next find an asymptotic expansion of $\tv{(z-H_p)^{-1}}$ using a thermodynamic limit procedure. 
 
We recall the Helffer-Sjostrand formula~\cite{HelSjo, Davies2}. For a self-adjoint operator $A$ and $\phi\in\mathcal S(\RR)$, we have
$$
\phi(A)=-\frac{1}{\pi}\int_\CC \frac{\partial \widetilde{\phi}}{\partial \overline{z}}(z) \frac{1}{z-A}\,dx\,dy,
$$
where $\widetilde{\phi}:\CC\rightarrow \CC$ is an appropriate complex extension of $\phi$ such that 

\renewcommand{\labelenumi}{(\roman{enumi})}
\begin{enumerate}
\item $\widetilde{\phi}\in \mathcal S(\set{z\in \CC,\; \av{\rm Im (z)}<1})$,
  \item for any $n\in\NN$ and $\alpha,\beta \geq 0$, one has
  \begin{equation}\label{eq:estimee_extention}
  \sup_{\av{y}<1}\mathcal N_{\alpha,\beta}\bra{x\mapsto \bra{\av{y}^{-n}\frac{\partial \widetilde{\phi}}{\partial \overline{z}}(x+iy)}}\leq C_{n,\alpha,\beta}\sup_{\beta'\leq n+\beta+2\atop \alpha'\leq \alpha}\mathcal N_{\alpha',\beta'}(\phi),
  \end{equation}
  where $\mathcal N_{\alpha,\beta}(\phi)=\sup_{x\in \RR}\av{x^\alpha\frac{\partial ^\beta \phi}{\partial x^\beta}}$.
\end{enumerate}
Hence, for $\phi\in\mathcal S(\RR)$, 
\begin{align*}
\langle n_p- n_0, \phi \rangle&=\int_\RR\phi(x) n_p(dx)-\int_\RR\phi(x) n_0(dx)=\tr\bra{\phi(H_p)-\phi(H_0)}\\
&=-\frac{1}{\pi}\tv{\int_\CC \frac{\partial \widetilde{\phi}}{\partial \overline{z}}(z) \bra{\frac{1}{z-H_p}-\frac{1}{z-H_0}} \,dx\,dy}.
\end{align*}
Besides, denoting by $V_p=V_{\nu_p}$, we have
\begin{align*}
\frac{1}{z-H_p}-\frac{1}{z-H_0}&=\frac{1}{z-H_0}V_p\frac{1}{z-H_p}
\end{align*}
Therefore, using the Kato-Seiler-Simon inequality~\eqref{eq:KSS} and Lemma~\ref{lemma:B(z)borné}, we obtain
\begin{align*}
 \av{\tv{\frac{1}{z-H_p}-\frac{1}{z-H_0}}}&\leq \norm{\1_{\Gamma}  \bra{-\Delta+1}^{-1}}_{\S_2}\norm{\bra{-\Delta+1}\frac{1}{z-H_0}}_{\B}\norm{V_p}_{L^\ii(\O\times \RRd)}\nonumber\\
&\times\norm{\frac{1}{z-H_p}\bra{-\Delta+1}}_{\B} \norm{ \bra{-\Delta+1}^{-1} \1_{\Gamma} }_{\S_2}\nonumber\\
&\leq C\bra{\frac{1+\av{z}}{\av{{\rm Im}(z)} }}^2\norm{V_p}_{L^\ii(\O\times \RRd)}.    
\end{align*}
By Fubini's theorem, we get
\begin{equation}\label{eq:diff_densite_detat_Helffer}
\langle n_p- n_0, \phi \rangle=-\frac{1}{\pi}\int_\CC \frac{\partial \widetilde{\phi}}{\partial \overline{z}}(z) \tv{\frac{1}{z-H_p}-\frac{1}{z-H_0}} \,dx\,dy.
\end{equation}
In the following, we find the asymptotic expansion of $\tv{(z-H_p)^{-1}-(z-H_0)^{-1}}$ as $p\rightarrow 0$ for $z\in\set{\CC\setminus\RR, \;\av{{\rm Im}(z)}\leq 1}$. 
To use a thermodynamic limit procedure, we consider, for each realization $\o\in \O$ and each box size $L\in2\NN+1$, the system with defects only in the box ${\Gamma} _L$, that is, we consider the defect distribution $\nu_{K_L(\o)}(x)$, with $K_L(\o)=\set{k\in\ZZd\cap {\Gamma} _L,\; q_k(\o)=1}$. For $K\subset\ZZd$, we recall the notation  $\nu_{K}= \sum_{k\in K}\chi(\cdot-k)$, $V_K=V_{\nu_K}=Y_m*(\rho_{\nu_K}-\nu_K)$ and $H_K=H_0+ V_{K}$ . By the proof of Theorem~\ref{th:limite_thermo}, we have, almost surely,
$$
\tr\bra{\1_{\Gamma} \bra{\frac{1}{z-H_p(\o)}-\frac{1}{z-H_{K_L(\o)}}}\1_{\Gamma} }\cvL 0.
$$
Besides, from~\eqref{eq:rho_controle_par_nu} and~\eqref{eq:Young}, it follows
\begin{align*}
\av{ \tr\bra{\1_{\Gamma} \bra{\frac{1}{z-H_p(\o)}-\frac{1}{z-H_{K_L(\o)}}}\1_{\Gamma} }}&\leq C \bra{\frac{1+\av{z}}{{\av{\rm Im}(z)}}}^2\norm{V_p(\o,\cdot)-V_{K_L}(\o,\cdot)}_{L^\ii} \\
&\leq C \bra{\frac{1+\av{z}}{\av{{\rm Im}(z) }}}^2\norm{\chi}_{L^2},
\end{align*}
The dominated converge theorem thus gives
$$
\EE\bra{\tr\bra{\1_{\Gamma} \bra{\frac{1}{z-H_p}-\frac{1}{z-H_{K_L}}}\1_{\Gamma} }}\cvL 0,
$$
and
\begin{align}\label{eq:limite_thermo}
 \tv{\frac{1}{z-H_p}-\frac{1}{z-H_0}}=\lim_{\cL} \EE\bra{\tr\bra{\1_{\Gamma} \bra{\frac{1}{z-H_{K_L}}-\frac{1}{z-H_0}}\1_{\Gamma} }}. 
\end{align}
Let $L\in 2\NN+1$ and $N=L^d$. As the random variable $\tr(\1_{\Gamma} ((z-H_{K_L})^{-1}-(z-H_0)^{-1})\1_{\Gamma}) $ depends only on the $N$ independent Bernoulli random variables $(q_k)_{k\in \ZZd\cap {\Gamma} _L}$, we have 
\begin{align*}
& \EE\bra{\tr\bra{\1_{\Gamma} \bra{\frac{1}{z-H_{K_L}}-\frac{1}{z-H_0}}\1_{\Gamma} }}\\
&\qquad\qquad\qquad  =\sum_{K\subset \ZZd\cap {\Gamma} _L}\PP(K_L(\o)=K)\tr\bra{\1_{\Gamma} \bra{\frac{1}{z-H_K}-\frac{1}{z-H_0}}\1_{\Gamma} }\\
&\qquad\qquad\qquad  =\sum_{n=0}^{N}p^n(1-p)^{N-n}\sum_{K\subset \ZZd\cap {\Gamma} _L\atop \av{K}=n} \tr\bra{\1_{\Gamma} \bra{\frac{1}{z-H_K}-\frac{1}{z-H_0}}\1_{\Gamma} }.
\end{align*}
Expanding the term $(1-p)^{N-n}$ as powers of $p$ and rearranging the sums, we obtain 
\begin{align}\label{eq:expansion_a_L_fixe}
& \EE\bra{\tr\bra{\1_{\Gamma} \bra{\frac{1}{z-H_{K_L}}-\frac{1}{z-H_0}}\1_{\Gamma} }}\nonumber\\
  &\qquad\qquad\qquad = \sum_{j=0}^{N}p^{j}\sum_{n=0}^{j}(-1)^{j-n}\binom{N-n}{j-n}\sum_{K\subset \ZZd\cap {\Gamma} _L\atop \av{K}=n} \tr\bra{\1_{\Gamma} \bra{\frac{1}{z-H_K}-\frac{1}{z-H_0}}\1_{\Gamma} }\nonumber\\
&\qquad\qquad\qquad = \sum_{j=0}^{J}a_{j,L}p^{j}+R_{J,L}(z,p),
\end{align}
where we have denoted the $j^{\rm th}$ order term by 
\begin{align*}
 a_{j,L}(z)&=\sum_{n=0}^{j}\binom{N-n}{j-n}\sum_{K\subset \ZZd\cap {\Gamma} _L\atop \av{K}=n} (-1)^{j-n}\tr\bra{\1_{\Gamma} \bra{\frac{1}{z-H_K}-\frac{1}{z-H_0}}\1_{\Gamma} }\\
&= \sum_{  K\subset \ZZd\cap {\Gamma} _L\atop \av{K}=j } \sum_{K'\subset K }(-1)^{\av{K\setminus K'}} \tr\bra{\1_{\Gamma} \bra{\frac{1}{z-H_{K'}}-\frac{1}{z-H_0}}\1_{\Gamma} }
\end{align*}
and the remainder of the series by 
\begin{align*}
 R_{J,L}(z,p)=\sum_{n=0}^{N}p^n(1-p)^{N-n}\sum_{K\subset \ZZd\cap {\Gamma} _L\atop \av{K}=n} \tr\bra{\1_{\Gamma} \bra{\frac{1}{z-H_K}-\frac{1}{z-H_0}}\1_{\Gamma} }-\sum_{j=0}^{J}a_{j,L}p^{j}.
\end{align*}
The result will now follow from the next proposition, whose proof is postponed until the end of the proof of the theorem.

\begin{proposition}[Estimates on $a_{j,L}$ and $R_{J,L}$]\label{prop:definition_termes_expansion+reste}
There exists $\alpha_c>0$ such that\\
$\bullet$ for $j \leq 2$, there exists $C\geq 0$ such that for any $\chi\in L^2(\RRd)$ satisfying ${\rm supp }(\chi)\subset \Gamma $ and $\norm{\chi}_{L^2}\leq \alpha_c$ and any $z\in \CC\setminus \RR$, 
\begin{equation}\label{eq:estimee_aj_prop}
\sum_{K\subset \ZZd\atop \av{K}=j}\av{ \sum_{K'\subset K }(-1)^{\av{K\setminus K'}} \tr\bra{\1_{\Gamma} \bra{\frac{1}{z-H_{K'}}-\frac{1}{z-H_0}}\1_{\Gamma} } }\leq C\norm{\chi}_{L^2}\bra{\frac{1+\av{z}}{\av{\rm Im (z)}}}^{j+1+jd}.
\end{equation}
$\bullet$ for $J\leq 2$, there exists $C\geq 0$ such that for any $\chi\in L^2(\RRd)$ satisfying ${\rm supp }(\chi)\subset \Gamma $ and $\norm{\chi}_{L^2}\leq \alpha_c$, $z\in \CC\setminus \RR$,  $p\in [0,1]$ and $L\in 2\NN+1$
\begin{equation}\label{eq:estimee-reste}
\av{R_{J,L}(z,p)}\leq C\norm{\chi}_{L^2}p^{J+1}\bra{\frac{1+\av{z}}{\av{\rm Im (z)}}}^{(J+2)(d+1)}.
\end{equation}
\end{proposition}

We deduce from Proposition~\ref{prop:definition_termes_expansion+reste} that for any  $j\leq 2$, and  $z\in \CC\setminus \RR$,  $a_{j,L}(z)$ converges as $\cL$ to 
\begin{align}\label{eq:estimee_aj}
a_j(z)= \sum_{  K\subset \ZZd \atop \av{K}=j } \sum_{K'\subset K }(-1)^{\av{K\setminus K'}} \tr\bra{\1_{\Gamma} \bra{\frac{1}{z-H_{K'}}-\frac{1}{z-H_0}}\1_{\Gamma} },
\end{align}
and that for any $J\leq 2$ and $p\in [0,1]$, $R_{J,L}(z,p)$ converges, up to extraction, as $\cL$ to $R_J(z,p)$, which satisfies 
\begin{align*}
 \av{R_J(z,p)}\leq C\norm{\chi}_{L^2} p^{J+1}\bra{\frac{1+\av{z}}{\av{\rm Im (z)}}}^{(J+2)(d+1)}.
\end{align*}
Passing to the limit as $\cL$ for this subsequence in~\eqref{eq:expansion_a_L_fixe} and in view of~\eqref{eq:limite_thermo}, we obtain
\begin{equation*}
\tv{\frac{1}{z-H_p}-\frac{1}{z-H_0}}=\sum_{j=1}^J a_j(z)p^j+p^{J+1}R_J(z,p).
\end{equation*}
Going back to~\eqref{eq:diff_densite_detat_Helffer}, we thus have
$$
\langle n_p- n_0, \phi \rangle=\sum_{i=1}^J\bra{-\frac{1}{\pi}\int_\CC \frac{\partial \widetilde{\phi}}{\partial \overline{z}}(z)a_j(z)\,dx\,dy}p^j- \bra{\frac{1}{\pi}\int_\CC \frac{\partial \widetilde{\phi}}{\partial \overline{z}}(z)R_J(z,p)\,dx\,dy}p^{J+1}.
$$
%
A simple calculation shows that 
\begin{align*}
 a_j(z)&=\frac{1}{j}\sum_{K\subset \ZZd\atop \av{K}=j,\; 0\in K}\sum_{k\in\ZZd}\sum_{K'\subset K+k }\bra{-1}^{\av{K\setminus K'}}\tr\bra{\1_{\Gamma}  \bra{ \frac{1}{z-H_{K'}}-\frac{1}{z-H_0} } \1_{\Gamma} }\\
&=\frac{1}{j}\sum_{K\subset \ZZd\atop \av{K}=j,\; 0\in K}\sum_{K'\subset K }\bra{-1}^{\av{K\setminus K'}}\tr\bra{  \frac{1}{z-H_{K'}}-\frac{1}{z-H_0}}.
\end{align*}
Therefore, by the dominated convergence theorem for series, we obtain
\begin{align*}
-\frac{1}{\pi}\int_\CC \frac{\partial \widetilde{\phi}}{\partial \overline{z}}(z)a_j(z)\,dx\,dy&= \frac{1}{j}\sum_{K\subset \ZZd\atop \av{K}=j,\; 0\in K}\sum_{K'\subset K }\bra{-1}^{\av{K\setminus K'}}\tr\bra{\phi(H_{K'})-\phi(H_0)}\\
&=\langle \mu_j,\phi\rangle.
\end{align*}
Moreover, using~\eqref{eq:estimee_extention},~\eqref{eq:estimee_aj_prop} and~\eqref{eq:estimee_aj}, we see that $\mu_j$
is a distribution of order at most $ j+3+jd$. Finally, $\phi\mapsto -\frac{1}{\pi}\int_\CC \frac{\partial \widetilde{\phi}}{\partial \overline{z}}(z)R_J(z,p)\,dx\,dy$ 
defines a distribution of order at most $ J+4+(J+2)d $ and satisfies
$$
\av{\frac{1}{\pi}\int_\CC \frac{\partial \widetilde{\phi}}{\partial \overline{z}}(z)R_J(z,p)\,dx\,dy}\leq C_{J}\sup_{ \beta \leq J+4+(J+2)d \atop \alpha\leq (J+3)(d +1) }\mathcal N_{\alpha,\beta}(\phi).
$$
This concludes the proof of Theorem~\ref{th:expansion}.
\end{proof}

To complete the proof of Theorem~\ref{th:expansion}, we need to prove Proposition~\ref{prop:definition_termes_expansion+reste}. We first state and prove Lemma~\ref{lemma:bounds} which will be useful in the proof of Proposition~\ref{prop:definition_termes_expansion+reste}.

\begin{lemma}\label{lemma:bounds}
Let $H=-\Delta+W$, with $W\in L^2_\unif(\RRd)$. Then, for any $\beta\in \NN$ and any Borel set $B\subset \RRd$, there exists $C\geq 0$ and $C'>0$ such that for any $z\in \CC\setminus \RR$ and any $\nu,\nu'\in L^2_c(\RRd)$ satisfying $\norm{\nu}_{L^2_\unif},\norm{\nu'}_{L^2_\unif}\leq \alpha_c$,  ${R}={\rm d} ({\rm supp}(\nu),0)\geq 1$, ${R'}={\rm d} ({\rm supp}(\nu'),0)\geq 1$, ${D}={\rm d} ({\rm supp}(\nu),{\rm supp}(\nu'))\geq 1$, we have

\begin{equation}\label{eq:loc1}
 \norm{\1_{\Gamma}  \frac{1}{z-H}V_\nu}_{\S_2}\leq C\frac{1+\av{z}}{\av{{\rm Im}(z)}}\bra{ e^{-C'\bra{\log {R} }^2}+e^{-C'c_2(z){R}}}\norm{\nu}_{L^2_\unif},
\end{equation}

\begin{equation}\label{eq:loc2}
\norm{V_{\nu\1_B} \frac{1}{z-H}\bra{V_{\nu+{\nu'}}-V_\nu}}_{\B}\leq
 \frac{C}{\av{{\rm Im}(z)}}\bra{\frac{1}{{D}^\beta}+e^{-C'c_2(z){D}}}\norm{\nu}_{L^2_\unif}\bra{\norm{\nu}_{L^2_\unif}+\norm{{\nu'}}_{L^2_\unif}},
\end{equation}

\begin{align}\label{eq:loc2bis}
 \norm{1_{\Gamma}  \frac{1}{z-H}\bra{V_{\nu+{\nu'}}-V_\nu}}_{\S_2}
&\leq C\frac{1+\av{z}}{\av{{\rm Im}(z)}}\com{\frac{1}{{D}^\beta}+e^{-C'c_2(z){D}}
+ e^{-C'\bra{\log {R'} }^2}+e^{-C'c_2(z){R'}} }\nonumber\\
&\times \bra{\norm{\nu}_{L^2_\unif}+\norm{{\nu'}}_{L^2_\unif}} 
\end{align}
and

\begin{align}\label{eq:loc3}
 \norm{\1_{\Gamma}  \frac{1}{z-H}\bra{V_{\nu+{\nu'}}-V_\nu-V_{\nu'}}}_{\S_2}
&\leq C\frac{1+\av{z}}{\av{{\rm Im}(z)}}\bra{\frac{1}{{D}^{\beta}}\bra{e^{-C'\bra{\log \widetilde{R} }^2}+e^{-C'c_2(z)\widetilde{R}}}}\nonumber\\
&\times \bra{\norm{\nu}_{L^2_\unif}+\norm{{\nu'}}_{L^2_\unif}},
\end{align}
where  $\widetilde{R}=\min\set{{R},R' }$, $c_2(z)={\rm d}(z, \sigma(H))/ (1+\av{z})$ and where the constants $C$ and $C'$ depend on $W$ only through its ${L^2_\unif}$-norm. 
\end{lemma}

\begin{proof}
Inequalities~\eqref{eq:loc1} -~\eqref{eq:loc3} follow from Lemmas~\ref{lemma:B(z)borné} and~\ref{lemma:CT}, Theorem~\ref{th:decay} and Proposition~\ref{prop:cor:loc}. 
For instance, for~\eqref{eq:loc1}, we first look at $V_\nu$ far from ${\Gamma} $. Using Lemma~\ref{lemma:CT}, we have
\begin{align*}
 \norm{\1_{\Gamma}  \frac{1}{z-H}\1_{\RRd\setminus B(0,\frac{{R}}{4})}V_\nu}_{\S_2}& \leq \frac{C}{\av{{\rm Im}(z)}}e^{-C'c_2(z){R}} \norm{V_\nu}_{L^\ii}\\
&\leq \frac{C}{\av{{\rm Im}(z)}}e^{-C'c_2(z){R}} \norm{\nu}_{L^2_\unif}.
\end{align*}
Near ${\Gamma} $, $V_\nu$ decays as ${R}$ gets large by Theorem~\ref{th:decay}. As ${\rm d}(B(0,\frac{{R}}{4}), {\rm supp}(\nu))\geq {R}/2 $, then, by~\eqref{eq:decayL2unif}, we have 
$$
\norm{\1_{B(0,\frac{{R}}{4})}V_\nu}_{L^\ii}\leq C \norm{\1_{B(0,\frac{{R}}{4})}V_\nu}_{H^2_\unif}\leq Ce^{-C'(\log{R})^2}\norm{\nu}_{L^2_\unif},
$$
where we have used that in dimension $d\leq 3$, $H^2_\unif(\RRd)\hookrightarrow L^\ii(\RRd)$. We next use Lemma~\ref{lemma:B(z)borné} and the Kato-Seiler-Simon inequality~\eqref{eq:KSS} to obtain
\begin{align*}
  \norm{\1_{\Gamma}  \frac{1}{z-H}\1_{B(0,\frac{{R}}{4})}V_\nu}_{\S_2}&\leq  \norm{\1_{\Gamma} \frac{1}{-\Delta+1} }_{\S_2}\norm{\bra{-\Delta+1}\frac{1}{z-H}}_{\B}\norm{\1_{B(0,\frac{{R}}{4})}V_\nu}_{L^\ii}\\
& \leq C \frac{1+\av{z}}{\av{{\rm Im}(z)}}e^{-C'(\log{R})^2}\norm{\nu}_{L^2_\unif},
\end{align*}
which concludes the proof of~\eqref{eq:loc1}. The proofs of~\eqref{eq:loc2},~\eqref{eq:loc2bis} and~\eqref{eq:loc3} use the same techniques; they can be found in~\cite{these}.
\end{proof}

We now prove Proposition~\ref{prop:definition_termes_expansion+reste}.

\begin{proof}[Proof of Proposition~\ref{prop:definition_termes_expansion+reste}]
Let $\alpha_c$ be the minimum of the constants $ \alpha_c$ defined in Theorems~\ref{th:point_fixe} and~\ref{th:decay} and Propositions~\ref{prop:loc} and~\ref{prop:cor:loc}. We assume that $\norm{\chi}_{L^2}\leq \alpha_c$. Throughout the proof, $\beta$ will denote an integer greater than $d+1$ whose value might change from one line to another and $C\geq 0$ and $C'>0$ constants that depend, in general, on $\beta$. For $z\in \CC\setminus \RR$,  we denote by $R_0(z)=(z-H_0)^{-1}$ and for any $K\subset \ZZd$, we set $R_K(z)=(z-H_K)^{-1}$. We omit the dependence on $z$ when there is no ambiguity. We also omit the $\norm{\chi}_{L^2}$ in our estimates.
Let $L\in 2\NN+1$ and denote by $N=L^d$. 

For $j=1$ and $K=\set{k}$, with $k\in \ZZd$, we have
\begin{align*}
\av{\tr\bra{ \1_{\Gamma}  \bra{ R_{\set{k}} -  R_0 }\1_{\Gamma} }}&=\av{\tr\bra{ \1_{\Gamma}  R_0V_{\set{k}}R_{\set{k}}\1_{\Gamma} }}\leq \norm{ \1_{\Gamma}  R_0V_{{\set{k}}}}_{\S_2}\norm{R_{\set{k}}\1_{\Gamma} }_{\S_2} .
 \end{align*}
Therefore, using~\eqref{eq:loc1} in Lemma~\ref{lemma:bounds}, we get
\begin{align*}
 \av{\tr\bra{ \1_{\Gamma}  R_{\set{k}}  V_{{\set{k}}} R_0 \1_{\Gamma} } }&\leq C\bra{\frac{1+\av{z}}{\av{{\rm Im}(z)}}}^2\bra{ e^{-\frac{C'}{2}\bra{\log \av{k} }^2}+e^{-\frac{C'}{2}c_2(z)\av{k}}}.
\end{align*}
Since the series $\sum_{k\in \ZZd}e^{-\lambda \av{k}}$, with $\lambda> 0$, is equivalent to $\int_\RRd e^{-\lambda \av{x}}dx=1/\lambda^d$, and for $z\in \set{z\in\CC,\;\av{{\rm Im}(z)}\leq 1 }$, it holds $1/c_2(z)\leq (1+\av{z})/\av{{\rm Im}(z)}$ and $ 1\leq (1+\av{z})/\av{{\rm Im}(z)}$, we deduce that 
the series $\sum_{k\in  \ZZd}\av{\tr\bra{ \1_{\Gamma}  R_{\set{k}}  V_{{\set{k}}} R_0 \1_{\Gamma} }  }$ is convergent and its sum satisfies
\begin{align*}
 \sum_{k\in  \ZZd}\av{\tr\bra{ \1_{\Gamma}  R_{\set{k}}  V_{{\set{k}}} R_0 \1_{\Gamma} }  }\leq C\bra{\frac{1+\av{z}}{\av{{\rm Im}(z)}}}^{2+d}.
\end{align*}
For $j=2$ and $K=\set{k,k'}$, with $k,k'\in \ZZd$, a straightforward calculation gives
\begin{align}\label{eq:TK2}
&\sum_{K'\subset K}(-1)^{\av{K\setminus K'}}\tr\bra{ \1_{\Gamma}   \bra{ R_{K'} -  R_0 }\1_{\Gamma}  }\nonumber\\
&\qquad\qquad\qquad\qquad= \tr\left(\1_{\Gamma}  R_0 (V_{\set{k,k'}}-V_{\set{k}}-V_{\set{k'}})R_{\set{k,k'}}\1_{\Gamma}  \right)\nonumber\\
&\qquad\qquad\qquad\qquad+ \tr\left(\1_{\Gamma}  R_0 V_{\set{k}}R_{\set{k}}\bra{V_{\set{k,k'}}-V_{\set{k}}}R_{\set{k,k'}}\1_{\Gamma}  \right)\nonumber\\
&\qquad\qquad\qquad\qquad+ \tr\left(\1_{\Gamma}  R_0V_{\set{k'}}R_{\set{k'}}\bra{V_{\set{k,k'}}-V_{\set{k'}}}R_{\set{k,k'}}\1_{\Gamma}  \right).
\end{align}
Using the inequality~\eqref{eq:loc3}, the first term of the RHS of~\eqref{eq:TK2} can be estimated by 
\begin{align*}
&\norm{ \1_{\Gamma} R_0 (V_{\set{k,k'}}-V_{\set{k}}-V_{\set{k'}}) }_{\S_2} \norm{ R_{\set{k,k'}}\1_{\Gamma}   }_{\S_2}\\
&\qquad\leq \bra{\frac{1+\av{z}}{\av{{\rm Im}(z)}}}^2\bra{\frac{C}{\av{k-k'}^{\beta}}\bra{e^{-C'\bra{\log \min\set{\av{k}, \av{k'}} }^2}+e^{-C'c_2(z)\min\set{\av{k}, \av{k'}}}}}.
\end{align*}
As to bound the second term of the RHS of~\eqref{eq:TK2}, it is bounded by 
\begin{align*}
&\norm{\1_{\Gamma} R_0 V_{\set{k}}}_{\S_2}\norm{ R_{\set{k}}\bra{V_{\set{k,k'}}-V_{\set{k}}}}_{\B}\norm{R_{\set{k,k'}}\1_{\Gamma} }_{\S_2}\\
&\qquad\qquad\qquad\qquad\leq C \bra{\frac{1+\av{z}}{\av{{\rm Im}(z)}}}^3 \bra{e^{-C'\bra{\log\av{k}}^2}+e^{-C'c_2(z)\av{k}}}.
\end{align*}
using~\eqref{eq:loc1}, and by 
\begin{align*}
&\norm{\1_{\Gamma} R_0}_{\S_2}\norm{V_{\set{k}}R_{\set{k}}\bra{V_{\set{k,k'}}-V_{\set{k}}}}_{\B}\norm{R_{\set{k,k'}}\1_{\Gamma} }_{\S_2}\\
&\qquad\qquad\qquad\qquad\leq C \bra{\frac{1+\av{z}}{\av{{\rm Im}(z)}}}^3 \bra{\frac{1}{\av{k-k'}^\beta}+e^{-C'c_2(z)\av{k-k'}}}.
\end{align*}
using~\eqref{eq:loc2}. Therefore 
\begin{align*}
 &\av{\tr\left(\1_{\Gamma} R_0 V_{\set{k}}R_{\set{k}}\bra{V_{\set{k,k'}}-V_{\set{k}}}R_{\set{k,k'}}\1_{\Gamma} \right)}\leq C \bra{\frac{1+\av{z}}{\av{{\rm Im}(z)}}}^3 \\
&\qquad\qquad\times\bra{e^{-C'\bra{\log\av{k}}^2}+e^{-C'c_2(z)\av{k}}}^\frac12
\bra{\frac{1}{\av{k-k'}^\beta}+e^{-C'c_2(z)\av{k-k'}}}^\frac12.
\end{align*}
We have the same bound for the third term of the RHS of~\eqref{eq:TK2}. Therefore, the series $\sum_{K\subset \ZZd\atop \av{K}=2}\av{ \sum_{K'\subset K}(-1)^{\av{K\setminus K'}}\tr\bra{ \1_{\Gamma}  \bra{ R_{K'} -  R_0 }\1_{\Gamma} }  }$ is convergent and its sum 
satisfies
\begin{align*}
 \sum_{K\subset \ZZd\atop \av{K}=2}\av{ \sum_{K'\subset K}(-1)^{\av{K\setminus K'}}\tr\bra{ \1_{\Gamma}  \bra{ R_{K'} -  R_0 }1_{\Gamma} } }\leq C \bra{\frac{\av{z}+1}{\av{\text{Im}(z)}}}^{3+2d}.
\end{align*}
We turn to the proof of the estimate on the remainder~\eqref{eq:estimee-reste}. Let $J\leq 2$ and $p\in \com{0,1}$. We first write $R_{J,L}(z,p)$ in the form of the expectancy of a binomial variable. Indeed, we have
\begin{align*}
 R_{J,L}(z,p)=\sum_{n=0}^Np^n(1-p)^{N-n}\sum_{K\subset \ZZd\cap \Gamma \atop \av{K}=n}f_{L,K}-\sum_{j=1}^Ja_{j,L}(z)p^j,
\end{align*}
where $f_{L,K}=\tr\bra{1_{\Gamma} \bra{R_K-R_0}1_{\Gamma} }$. Rearranging all the terms (see~\cite{these} for details), we obtain
\begin{align*}
 R_{J,L}(z,p)
&=p^{J+1}\EE\bra{g_{J,L}\bra{Y_L+J+1,z}},
\end{align*}
where $Y_L$ is a random variable of binomial distribution of parameters $p$ and $N-J-1$ and $g_{J,L}(\cdot,z): \set{J+1,\cdots, N}\rightarrow \RR$ is defined by 
\begin{align*}
&g_{J,L}(n,z)= \binom{N-J-1}{n-J-1}^{-1}\sum_{K\subset \ZZd\cap {\Gamma} _L\atop \av{K}=n}\tr\bra{1_{\Gamma} \bra{R_K(z)-R_0(z)}1_{\Gamma} }\\
&-  \binom{N-J-1}{n-J-1}^{-1}\sum_{K\subset \ZZd\cap {\Gamma} _L\atop \av{K}=n}\sum_{ K'\subset K \atop \av{K'}\leq J} \sum_{K''\subset K'}(-1)^{K'\setminus K'' } \tr\bra{1_{\Gamma} \bra{R_{K''}(z)-R_0}1_{\Gamma} }.
\end{align*}
In order to prove~\eqref{eq:estimee-reste}, it is therefore sufficient to show that there exists $C\geq 0$ such that for any $L\in 2\NN+1$ and $J+1\leq n\leq N$,
\begin{align*}
\av{ g_{J,L}(n,z)}\leq C \bra{\frac{1+\av{z}}{\av{ {\rm Im}(z) }}}^{J+2+ (J+1)d}.
\end{align*}
It is sufficient to prove the above inequality for $J=2$. Let $J+1\leq n\leq N$ and consider a configuration $K\subset\ZZd\cap {\Gamma} _L$ such that $\av{K}=n$. 
A straightforward calculation shows that  
$$
g_{J,L}(n,z)= \binom{N-J-1}{n-J-1}^{-1}\sum_{K\subset \ZZd\cap {\Gamma} _L\atop \av{K}=n+J+1 }\tr\bra{ 1_{\Gamma}  R_0 \bra{P_{1,K}-P_{2,K}}R_K  1_{\Gamma} }
$$
where 
$$
P_{1,K}=V_K-  \sum_{k\in K }V_{\set{k}} -\sum_{k,k'\in K\atop k\neq k'} \bra{ V_{\set{k,k'}}-V_{\set{k}} -V_{\set{k'}} }  
$$
and 
\begin{align*}
 P_{2,K}= \sum_{k\in K }V_{\set{k}}R_{k}& \bra{V_K-  V_{\set{k}}  }+ \sum_{\set{k,k'}\subset K}\Big(  V_{\set{k,k'}} R_{\set{k,k'}}  \bra{V_K-  V_{\set{k,k'}}  }  \\
&-  V_{\set{k}} R_{\set{k}}  \bra{V_K-  V_{\set{k}}  }-  V_{\set{k'}} R_{\set{k'}}  \bra{V_K-  V_{\set{k'}}  } \Big).
\end{align*}
Besides
\begin{align*}
P_{1,K}&=\sum_{r\in \ZZd}\1_{{\Gamma} +r}P_{1,K}.
\end{align*}
For each $r\in \ZZd$, we split $\1_{{\Gamma} +r}P_{1,K} $ into two $r$-dependent quantities: a part involving the defect in $k_0=\arg\inf_{k\in K}\av{k-r}$ and the rest. We denote by
$$
A_{K,k_0}=V_K- V_{\set{k_0}}- \sum_{k\in K\setminus\set{ k_0} }\bra{V_{\set{k,k_0}}-V_{\set{k_0}}}$$
and 
$$ B_{K,k_0}=\sum_{\set{k,k'}\subset K\setminus \set{k_0}}  \bra{V_{\set{k,k'}}-V_{\set{k}} -V_{\set{k'}}}.
$$
Then  
\begin{equation*}
P_{1,K}=\sum_{r\in\ZZd} \1_{{\Gamma} +r} A_{K,k_0}-\sum_{r\in\ZZd}\1_{{\Gamma} +r}B_{K,k_0}.
\end{equation*}
We have thus split $g_{J,L}(n,z)$ into three parts

\begin{align}\label{eq:g_Lnz_split_into_three_part}
 g_{J,L}(n,z)&= 
\binom{N-J-1}{n-J-1}^{-1}\sum_{K\subset\ZZd\cap {\Gamma} _L\atop \av{K}=n}\tr\bra{1_{\Gamma} R_0  \sum_{r\in\ZZd}\1_{{\Gamma} +r}A_{K,k_0}R_K1_{\Gamma} }\nonumber\\
&\quad-\binom{N-J-1}{n-J-1}^{-1}\sum_{K\subset\ZZd\cap {\Gamma} _L\atop \av{K}=n}\tr\bra{\sum_{r\in\ZZd}\1_{{\Gamma} +r}B_{K,k_0}R_K1_{\Gamma} }\nonumber\\
&\quad + \binom{N-J-1}{n-J-1}^{-1}\sum_{K\subset\ZZd\cap {\Gamma} _L\atop \av{K}=n}\tr\bra{P_{2,K} R_K1_{\Gamma} }
\end{align}
that we will bound successively. We start by the first term. Let $r\in \ZZd$ and denote by $k_1=\arg\inf_{k\in K\setminus\set{ k_0} }{\rm d}\bra{k,\set{r,k_0}}$. We introduce 
$$
\ell_0(K,r)=\av{r-k_0}, \quad
\ell_1(K,r)={\rm d}\bra{K\setminus \set{k_0},\set{r,k_0}}
$$
and
$$
\ell_2(K,r)={\rm d}\bra{K\setminus\set{ k_0, k_1},\set{r,k_0,k_1 }}. 
$$
When there is no ambiguity, we omit to note the dependence of these quantities on $K$ and $r$. By Theorem~\ref{th:decay}, we first have
\begin{align}\label{eq:AK1}
\norm{1_{{\Gamma} +r}\bra{V_K- V_{\set{k_0}}}}_{L^\ii}&\leq \norm{1_{{\Gamma} +r}V_K}_{L^\ii} +\norm{1_{{\Gamma} +r}V_{\set{k_0}}}_{L^\ii}
\leq \frac{C}{\bra{\ell _0+1}^\beta}.
\end{align}
We now want to control $\norm{1_{{\Gamma} +r}\bra{V_K- V_{\set{k_0}}}}_{L^\ii}$ by $1/(\ell _1+1)^\beta$. 
If $\ell _0< \ell _1/4^\beta $ (see Figure~\ref{fig:shema_ll}), then by Proposition~\ref{prop:cor:loc}, we have
\begin{equation}\label{eq:AK2}
\quad \norm{1_{{\Gamma} +r}\bra{V_K- V_{\set{k_0}}}}_{L^\ii}\leq \frac{C}{(\ell _1+1)^\beta}.
\end{equation}
\begin{figure}[!h]
\centering
\psfrag{k0}{\rotatebox{270}{\hspace{-0.3cm}$k_0$}} \psfrag{k1}{\rotatebox{270}{$k_1$}} \psfrag{r}{\rotatebox{270}{$r$}} \psfrag{l0}{\rotatebox{270}{\hspace{-0.2cm}$\ell_0$}} \psfrag{l1}{\rotatebox{270}{$\ell _1$}} \psfrag{l1/a}{\rotatebox{270}{$~~\ell _1/4^\beta$}}
\includegraphics[scale=0.6,angle=90]{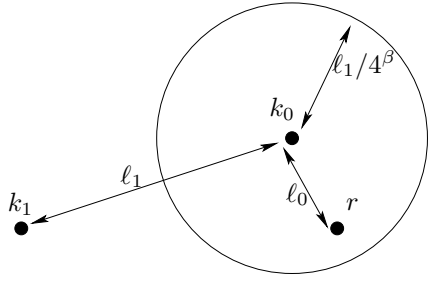}
\caption{A configuration of $r$, $k_0$ and $k_1$ where $\ell _0\leq \ell _1/4^\beta$ used in the proof of Lemma~\ref{lemma:bounds}.}
\label{fig:shema_ll}
\end{figure}
If $\ell _0\geq \ell _1/4^\beta$, then~\eqref{eq:AK1} gives 
\begin{align}\label{eq:AK2bis}
\norm{1_{{\Gamma} +r}\bra{V_K- V_{\set{k_0}}}}_{L^\ii}&\leq \frac{C}{\bra{\ell _0+1}^\beta}\leq \frac{C}{\bra{\ell _1+1}^\beta}.
\end{align}
Therefore, by~\eqref{eq:AK1},~\eqref{eq:AK2} and~\eqref{eq:AK2bis},
\begin{align}\label{eq:AK2ter}
 \norm{1_{{\Gamma} +r}\bra{V_K- V_{\set{k_0}}}}_{L^\ii}& =\norm{1_{{\Gamma} +r}\bra{V_K- V_{\set{k_0}}}}_{L^\ii}^\frac12\times \norm{1_{{\Gamma} +r}\bra{V_K- V_{\set{k_0}}}}_{L^\ii}^\frac12\nonumber\\
&\leq \frac{C}{\bra{\ell _0+1}^\frac{\beta}{2}\bra{\ell _1+1}^\frac{\beta}{2}}
\end{align}
We proceed similarly for the remaining term of $A_{K,k_0}$.  First, as~\eqref{eq:AK2ter} holds for any $\beta\geq 0$ and any $K\ni k_0$, then we have for any $k\in K\setminus \set{k_0}$
\begin{align}\label{eq:AK3bis}
 \norm{1_{{\Gamma} +r}\bra{V_{\set{k_0,k}}- V_{\set{k_0}}}}_{L^\ii}\leq \frac{C}{(\ell _1+1)^\beta (\ell _0+1)^\beta}.
\end{align}
Next, if $\ell _0< \av{k-k_0}/4^\beta $, then by Proposition~\ref{prop:cor:loc}, we have 
\begin{align*}
 \norm{1_{{\Gamma} +r}\bra{V_{\set{k_0,k}}- V_{\set{k_0}}}}_{L^\ii}\leq \frac{C}{\av{k_0-k}^\beta}.
\end{align*}
Otherwise, by~\eqref{eq:AK3bis}
\begin{align*}
 \norm{1_{{\Gamma} +r}\bra{V_{\set{k_0,k}}- V_{\set{k_0}}}}_{L^\ii}\leq \frac{C}{(\ell _0+1)^\beta}\leq \frac{C}{(\av{k-k_0}+1)^\beta}.
\end{align*}
Therefore, reasoning as in~\eqref{eq:AK2ter}, we have for $\beta$ large enough
\begin{align}\label{eq:AK3}
\norm{\sum_{k\in K\setminus\set{ k_0} }1_{{\Gamma} +r}\bra{V_{\set{k_0,k}}- V_{\set{k_0}}}}_{L^\ii}
& \leq\sum_{k\in K\setminus\set{ k_0} }\bra{\norm{1_{{\Gamma} +r}\bra{V_{\set{k_0,k}}- V_{\set{k_0}}}}_{L^\ii}^\frac12}^2 \nonumber\\
&\leq \frac{C}{(\ell _1+1)^\frac{\beta}{2} (\ell _0+1)^\frac{\beta}{2}}\sum_{k\in K\setminus\set{ k_0} } \frac{1}{\av{k_0-k}^\frac{\beta}{2}}\nonumber\\
&\leq \frac{C}{(\ell _1+1)^\frac{\beta}{2}  (\ell _0+1)^\frac{\beta}{2} }.
\end{align}
As~\eqref{eq:AK3} and~\eqref{eq:AK2ter} holds for any $\beta\geq 0$, then by the definition of $A_{K,k_0}$, we obtain
\begin{equation}\label{eq:AK3ter}
\av{\1_{{\Gamma} +r}A_{K,k_0}}\leq \frac{C}{(\ell _1+1)^\beta (\ell _0+1)^\beta}.
\end{equation}
To control $A_{K,k_0}$ by $1/\ell _2^\beta$, we rearrange the terms of $A_{K,k_0}$ as follows
$$
A_{K,k_0}=V_K- V_{\set{k_0,k_1}}- \sum_{K\setminus\set{ k_0, k_1} }\bra{V_{\set{k,k_0}}-V_{\set{k_0}}}.
$$
By Proposition~\ref{prop:cor:loc}, we thus have
\begin{equation}\label{eq:AK4}
\av{\1_{{\Gamma} +r}A_{K,k_0}}\leq \frac{C}{\ell_2^\beta}+ \sum_{k\in K\setminus\set{k_0,k_1} }\frac{C}{\av{k-k_0}^\beta}\leq  \frac{C}{\ell _2^\beta}.
\end{equation}
As~\eqref{eq:AK3ter} and~\eqref{eq:AK4} hold for any $\beta$, then reasoning as in the proof of~\eqref{eq:AK2ter} we have
\begin{align*}
\tr\bra{ \av{1_{\Gamma}  R_0 \1_{{\Gamma} +r} A_{K,k_0} R_K  1_{\Gamma} }}
& \leq  \norm{ 1_{\Gamma}  R_0\1_{{\Gamma} +r}}_{\S_2} \norm{  \1_{{\Gamma} +r} A_{K,k_0} }_{L^\ii} \norm{R_K  1_{\Gamma} }_{\S_2}\\
& \leq C\frac{e^{-C'c_2(z)\av{r}}}{\av{{\rm Im}(z)}}\frac{1}{(\ell _0+1)^\beta(\ell _1+1)^\beta\ell_2^\beta}\frac{1+\av{z}}{\av{{\rm Im}(z)}}.
\end{align*}
Therefore  $\sum_{r\in\ZZd}\tr\bra{ \av{1_{\Gamma}  R_0 \1_{{\Gamma} +r} A_{K,k_0} R_K  1_{\Gamma} }}$ is a convergent series. By Fubini's Theorem, we thus have
\begin{align*}
\hspace{-0.5cm} \sum_{ K\subset \ZZd\cap{\Gamma} _L\atop \av{K}=n}\tr\bra{ 1_{\Gamma}  R_0 \bra{\sum_{r\in\ZZd} \1_{{\Gamma} +r}A_{K,k_0} }R_K  1_{\Gamma} }
=\sum_{r\in\ZZd}\sum_{ K\subset \ZZd\cap{\Gamma} _L\atop \av{K}=n}\tr\bra{ 1_{\Gamma}  R_0  \1_{{\Gamma} +r}A_{K,k_0} R_K  1_{\Gamma} }.
\end{align*}
To perform the sum over the configurations $K\in \set{ K\subset \ZZd\cap{\Gamma} _L,\; \av{K}=n}$, we classify these configurations depending on the value of $\ell_i(r,K)$, $i\in\set{0,1,2}$: 
\begin{align*}
&\av{ \sum_{ K\subset \ZZd\cap{\Gamma} _L\atop \av{K}=n}\tr\bra{ 1_{\Gamma}  R_0\bra{ \sum_{r\in\ZZd} \1_{{\Gamma} +r}A_{K,k_0}} R_K  1_{\Gamma} }}\\
&\quad\quad\leq  \sum_{r\in\ZZd} \sum_{L_0,L_1,L_2=0}^{\sqrt{d}L}  \sum_{ K\subset \ZZd\cap{\Gamma} _L,\; \av{K}=n\atop L_i\leq \ell _i(K,r)< L_i+1} 
Ce^{-C'c_2(z)\av{r}} \bra{\frac{1+\av{z}}{\av{{\rm Im}(z)}}}^2  \frac{1}{\Pi_{i=0}^2(L_i+1)^\beta}.\\
&\quad\quad\leq \sum_{r\in \ZZd}\sum_{L_0,L_1,L_2=0}^{\sqrt{d}L} Ce^{-C'c_2(z)\av{r}}\bra{\frac{1+\av{z}}{\av{{\rm Im}(z)}}}^2 \frac{N_{L,n,r}(L_0,L_1,L_2)}{\Pi_{i=0}^2(L_i+1)^\beta},
\end{align*}
where $N_{L,n,r}(L_0,L_1,L_2)$ is the number of configurations $K\subset \ZZd\cap {\Gamma} _L$ such that $\av{K}=n$ and $L_i\leq \ell _i(K,r)< L_i+1$ for $i\in \set{0,1,2}$. This number can be estimated by the asymptotic value $C\binom{N-3}{n-3}\prod_{i=0}^2L_i^{d-1}$ when $N\rightarrow \ii$. Therefore, taking $\beta$ large enough, we obtain that the first term of the RHS of~\eqref{eq:g_Lnz_split_into_three_part} is bounded by 
\begin{align*}
\sum_{r\in \ZZd}\sum_{L_0,L_1,L_2=1}^L Ce^{-C'c_2(z)\av{r}}\bra{\frac{1+\av{z}}{\av{{\rm Im}(z)}}}^2 \frac{1}{\Pi_{i=0}^2(L_i+1)^\beta}\leq C\bra{\frac{1+\av{z}}{\av{{\rm Im}(z)}}}^{2+d}.
\end{align*}
With the same techniques, we find that the second and third terms of the RHS of~\eqref{eq:g_Lnz_split_into_three_part} are respectively bounded by 
$$
C \bra{\frac{1+\av{z}}{\av{{\rm Im}(z)}}}^{2+d}\quad\text{and}\quad C \bra{\frac{1+\av{z}}{\av{{\rm Im}(z)}}}^{4+4d},
$$
which concludes the proof of the proposition.

\end{proof}

\appendix

\bibliographystyle{siam}
\bibliography{fichierb}

\end{document}